\documentclass[12pt]{article}

\usepackage{bbm}
\usepackage{setspace,graphicx,amsmath,amsfonts,amssymb,amsthm,versionPO}
\usepackage{thmtools}

\usepackage{marginnote,datetime,rotating,fancyvrb}
\usepackage[authoryear,round,longnamesfirst]{natbib}
\usepackage{doi}
\usepackage{relsize}
\usepackage{lipsum}
\usepackage[toc,page]{appendix}
\usepackage{tikz}
\usepackage{physics}
\usepackage{xspace}
\usepackage[inline,shortlabels]{enumitem}
\setlist[enumerate,1]{label=(\arabic*)}

\newcommand\blfootnote[1]{%
  \begingroup
  \renewcommand\thefootnote{}\footnote{#1}%
  \addtocounter{footnote}{-1}%
  \endgroup
}
\usdate

\makeatother

\makeatletter

\excludeversion{notes}		
\includeversion{links}          

\iflinks{}{\hypersetup{draft=true}}

\ifnotes{%
\usepackage[margin=1in,paperwidth=10in,right=2.5in]{geometry}%
\usepackage[textwidth=1.4in,shadow,colorinlistoftodos]{todonotes}%
}{%
\usepackage[margin=1in]{geometry}%
\usepackage[disable]{todonotes}%
}

\makeatletter\let\chapter\@undefined\makeatother 

\makeatletter
\renewcommand{\paragraph}{%
  \@startsection{paragraph}{4}%
  {\z@}{1.25ex \@plus 1ex \@minus .2ex}{-1em}%
  {\normalfont\normalsize\bfseries}%
}

\let\oldparagraph=\paragraph
\renewcommand\paragraph[1]{\oldparagraph{#1.}}

\makeatother

\newtheorem{thm}{Theorem}

\newtheorem{prop}{Proposition}

\newtheorem{lem}{Lemma}
\newtheorem{cor}{Corollary}

\newtheoremstyle{definition_space} 
  {\topsep} 
  {\topsep} 
  {} 
  {} 
  {\bfseries} 
  {.} 
  {.5em} 
  {} 

\theoremstyle{definition_space}
\newtheorem{defn}{Definition}
\newtheorem{rem}{Remark}

\newenvironment{remtri}
  {\pushQED{\qed}\rem}
  {\popQED\endrem}

\newcommand{\E}{\mathbb{E}}
\newcommand{\R}{\mathbb{R}}

\DeclareMathOperator{\sgn}{sgn}

\title{Robust Comparative Statics for the Elasticity of Intertemporal Substitution}
    
\author{Joel P. Flynn\thanks{MIT Department of Economics, 50 Memorial Drive, Cambridge, MA 02142. Email: \href{mailto:jpflynn@mit.edu}{jpflynn@mit.edu}.}\\ MIT \\\and Lawrence D. W. Schmidt\thanks{MIT Sloan School of Management, 100 Main Street, Cambridge, MA 02142. Email: \href{mailto:ldws@mit.edu}{ldws@mit.edu}.} \\ MIT \\\and Alexis Akira Toda\thanks{University of California San Diego Department of Economics, 9500 Gilman Drive \#0508 
La Jolla, CA 92093-0508. Email: \href{mailto:atoda@ucsd.edu}{atoda@ucsd.edu}.}\\ UCSD \\ }
    
\date{\today}

\usepackage{hyperref}
\hypersetup{unicode=true,
 breaklinks=false,pdfborder={0 0 1},
 colorlinks,linkcolor={red!50!black},citecolor={blue!50!black},urlcolor={blue!80!black}}

\begin{document}
\onehalfspacing

\clearpage
\maketitle
\thispagestyle{empty}

\begin{abstract}
\noindent We study a general class of consumption-savings problems with recursive preferences. We characterize the sign of the consumption response to arbitrary shocks in terms of the product of two sufficient statistics: the elasticity of intertemporal substitution between contemporaneous consumption and continuation utility (EIS), and the relative elasticity of the marginal value of wealth (REMV). Under homotheticity, the REMV always equals one, so the propensity of the agent to save or dis-save is always signed by the relationship of the EIS with unity. We apply our results to derive comparative statics in classical problems of portfolio allocation, consumption-savings with income risk, and entrepreneurial investment. Our results suggest empirical identification strategies for both the value of the EIS and its relationship with unity.
\end{abstract}

\blfootnote{Previous versions of this paper were circulated with the titles ``Do You Save More or Less in Response to Bad News? A New Identification of the Elasticity of Intertemporal Substitution'' and ``Bad News and Robust Comparative Statics for the Elasticity of Intertemporal Substitution'' \citep{schmidt-toda-EIS}. We thank Tomas Caravello, David Cashin, Roberto Corrao, Henrik Kleven, Stephen Morris, Jonathan Parker, Alp Simsek, Gianluca Violante, and seminar participants at University of Chicago, Hitotsubashi University International Corporate Strategy, University of Virginia Darden School of Business, and 2015 Econometric Society World Congress for comments and feedback.}

\maketitle
\setcounter{page}{0}

\section{Introduction}
\label{sec:intro}
A growing body of theoretical literature assumes that investors have recursive preferences, particularly the constant relative risk aversion, constant elasticity of intertemporal substitution specification studied in \cite{EpsteinZin1989} and \cite{Weil1989}. This utility function allows two different parameters (relative risk aversion and elasticity of intertemporal substitution) to separately govern the attitude towards risky gambles and the willingness to smooth consumption over time, which are mechanically linked for additively separable constant relative risk aversion (CRRA) preferences.

In applied theoretical models with Epstein-Zin preferences and generalizations thereof, there remains a considerable debate with respect to ``reasonable'' choices for the elasticity of intertemporal substitution (EIS), which is conventionally defined in terms of a comparative static: the percentage change in consumption growth induced by a one percent increase in the rate of return on investment. This debate persists in part because empirical estimates of EIS vary considerably from being larger than one to significantly negative.\footnote{See \cite{HHIR2015} for a review.}

This lack of consensus regarding the magnitude of EIS is troubling since the relationship between EIS and unity plays a central role in affecting the dynamics of many theoretical models in both quantitative and \textit{qualitative} terms. For example, in the \cite{bansal-yaron2004} long-run risk model, when $\text{EIS}>1$, investors are willing to pay a premium to hedge against lower future economic growth rates. The wealth-consumption ratio is pro-cyclical, the equity premium is high, and the risk-free rate is low and stable. Setting $\text{EIS}<1$ changes many basic intuitions for the model and often reverses each of these properties.\footnote{Similar changes occur to asset prices and quantity dynamics in production-based models \citep{kaltenbrunner-lochstoer2010,croce2014}.} Moreover, \cite{KaplanViolante2014} find that $\text{EIS}>1$ is crucial for getting households to hold large illiquid positions and thus to quantitatively match the consumption response to tax rebates. Assuming $\text{EIS}>1$ also has striking implications for the response of asset prices to changes in uncertainty.\footnote{In endowment economies, \cite{bansal-yaron2004} and \cite{Barro2009AER} find that asset prices fall in response to increased volatility and disaster risk when $\text{EIS}>1$ and rise otherwise. \cite{drechsler-yaron2011} and \cite{ditella2017} argue that $\text{EIS}>1$ is crucial for explaining the variance risk premium and balance sheet recessions, respectively. 
}

This paper develops robust comparative statics for an investor's optimal consumption-savings decision in a general portfolio problem with recursive preferences over contemporaneous consumption and a certainty equivalence functional of their continuation utility. We generalize the definition of EIS to this setting, where, in contrast to the more conventional definition of EIS (and target estimand in the empirical literature), our definition depends only on preferences and makes no specific assumptions on budget constraints and investment opportunities. In the Epstein-Zin case (which nests the standard CRRA expected-utility model), our definition always corresponds with the structural parameter. By contrast, the more conventional definition may or may not correspond with the structural parameter depending on additional details of the budget constraint of the consumption-savings-investment problem.\footnote{For example, consumers may face hard borrowing constraints which bind with positive probability. This induces a wedge in the consumption-savings problem that leads to a difference between the ``standard'' measure of EIS and ours, even when preferences themselves are homothetic.} Our specification allows for essentially unrestricted flexibility in risk tolerance, impatience, willingness to substitute over time, ambiguity aversion, riskiness of returns, investment opportunities, and both state and time variation in all of these factors.

We provide three main results, which establish a tight link between our general notion of the EIS and the consumption responses to shocks to continuation values. First, we show that the sign of the consumption response to shocks is characterized by the relationship with unity of the product of two sufficient statistics: the EIS, capturing willingness to substitute consumption over time; and the relative elasticity of the marginal value of wealth (REMV), which is the ratio between the elasticity of the marginal value of wealth and the elasticity of the value of wealth. This statistic captures the size of wealth effects induced by the shock. Hence, consumption increases in response to a positive shock to continuation values if and only if $\text{EIS} \times \text{REMV} \le 1$. 

Second, we show that if the agent's preferences are homothetic, then the REMV is identically equal to one. Thus, the signs of consumption responses with homothetic preferences are characterized precisely by the relationship of the EIS with unity.

Third, using techniques from the literature on monotone comparative statics \citep{MilgromShannon1994}, we provide general sufficient conditions on the agent's preferences for any possible optimal consumption function to be globally increasing or decreasing with respect to arbitrary shifters of continuation utility. Our sufficient condition shows that if the product of suitably generalized global counterparts to the EIS and REMV are always less than unity, then consumption is increasing in shifters that increase continuation value.

The intuition for these results is best exemplified in a simple two-period setting without risk, which forms the first section of the paper. Concretely, consider an agent choosing whether to consume today or tomorrow in the presence of a risk-free asset that they can freely trade. The sign of the consumption response to an increase in the risk-free rate depends on two factors. First, if the interest rate increases, then the opportunity cost of contemporaneous consumption increases, which induces the consumer to wish to save more and increase continuation utility via a substitution effect. If continuation utility is a gross substitute for consumption today (\textit{i.e.,} $\text{EIS}>1$), then this effect pushes the agent to substitute consumption today for more continuation utility. Conversely, if continuation utility is a gross complement for consumption today (\textit{i.e.,} $\text{EIS}<1$), then greater continuation utility crowds in more consumption today. Second, if the agent earns future endowment income, then the change in interest rates reduces the value of their future endowment. This reduces the magnitude of the wealth effect induced by changes in interest rates and makes the agent more predisposed to cut consumption (\textit{i.e.,} $\text{REMV}>1$). Thus, consumption falls under the less stringent condition that $\text{EIS}\times\text{REMV}>1$. Our general results clarify that this simple intuition is fully general: one need only work out the EIS and the REMV to sign consumption responses to any shock.

We apply our results to understand how consumption responds to various shocks in three applications to
\begin{enumerate*}
\item portfolio allocation,
\item consumption-savings with labor income risk, and
\item entrepreneurial investment.
\end{enumerate*}
To operationalize our theoretical results, we show in our three settings that continuation values are adversely affected by:
\begin{enumerate*}
\item increases in risk aversion, reduced investment opportunities, lower returns to investment, riskier returns to investment, diminished continuation value of consumption, and increased ambiguity aversion;
\item lower labor income, increased income risk, and reduced opportunities to hedge income risk;
\item less productive production technology, higher rental rates for capital and labor, higher depreciation rates, riskier depreciation rates, and higher capital tax rates.
\end{enumerate*}
Thus, in each case, our general theoretical results can be applied directly to show that consumption decreases in response to these changes if and only if $\text{EIS} \times \text{REMV} \le 1$. We moreover provide sufficient conditions for each environment to be homothetic, in which case consumption decreases if and only if $\text{EIS}\le 1$.

Finally, we provide practical guidance on how to leverage these comparative statics to test whether $\text{EIS} \gtreqless 1$. Concretely, by characterizing the sign of the consumption response to shocks in terms of the relationship of the EIS with unity, we achieve sign-identification of $\text{EIS}-1$ by observing the sign of the consumption response to shocks. Thus, our results may allow empirical researchers to exploit variations in multiple variables beyond risk-free returns to estimate the sign of $\text{EIS}-1$ in a model-free manner simply by observing whether an agent consumes more or less. This is important because the standard estimand of the elasticity of the growth rate of consumption to risk-free rates (the standard empirical measure of the EIS) need not coincide with the structural definition of the EIS under realistic frictions, such as borrowing constraints. Moreover, while these tests only partially identify EIS (\textit{i.e.,} its relationship with unity), with additional structural assumptions one can use our formulas for consumption responses to point-identify EIS.

\paragraph{Related Literature} Some of the theoretical implications of the relationship between EIS and one have appeared in the literature. For example, working with CRRA preferences, many classic papers have found that riskier environments increase or decrease savings depending on whether relative risk aversion is greater or less than one.\footnote{See, for example, \citet[Section 6(ii)]{Phelps1962}, \citet[p.~161]{levhari-srinivasan1969}, \citet[p.~254]{merton1969}, \citet[p.~358]{Sandmo1970}, and \citet[p.~70]{RothschildStiglitz1971Risk2}.} Our results are considerably more general as they place no parametric structure on the environment. Moreover, by working with recursive preferences, we clarify that these results are driven exclusively by EIS, not risk aversion, extending an intuition developed in \cite{Weil1993} and the approximate solutions of \cite{campbell1993} and \cite{CampbellGiglioPokTurleyI2018}.\footnote{\cite{Weil1993} considers an optimal consumption-savings problem with recursive preferences that exhibit constant elasticity of intertemporal substitution (CEIS) and constant absolute risk aversion (CARA). The relevant comparative statics appears in Section 2.5 of that paper, where he assumes income is independent and identically distributed (IID).} We allow for considerably more flexibility in how risk aversion, time discounting, and even future EIS evolve over time, and also allow for preferences that incorporate ambiguity aversion \citep[as modeled in, \textit{e.g.},][]{epstein-schneider2003,hayashi2005,hayashi-miao2011}, realistic life cycle features, and various kinds of investment opportunities.

Our paper is also related to \cite{Epstein1988}, who studies the asset pricing implications of recursive preferences in a representative-agent endowment economy. He derives comparative statics results, which depend on the magnitude of EIS relative to unity. Our paper is different because we focus on the response of consumption behavior to arbitrary shocks and impose much weaker restrictions, allowing for general recursive preferences and stochastic processes. More recently, in independent work, \cite{IachanNevovSimsek2021} show that expanding portfolio choice (which they refer to as financial innovation) increases savings if EIS is greater than one. Our analysis is complementary since we put little structure on the model and consider many other types of comparative statics, whereas they focus on one channel but also study general equilibrium implications. An advantage of the level of generality considered here is that, in addition to allowing for more flexibility than existing theoretical results, we can summarize many key predictions about savings behavior with recursive preferences in a simple, self-contained way that also has general implications for how applied researchers can set- and point-identify EIS from different shocks in different settings.

\paragraph{Outline} The paper proceeds as follows. Section \ref{sec: example} develops a simple two-period example to illustrate our main results. Section \ref{sec:results} describes our general model and main results. Section \ref{sec:applications} applies our results to problems of portfolio allocation, consumption-savings, and entrepreneurial investment. Section \ref{sec:implications} describes the implications of these results for identification and estimation of EIS. Section \ref{sec:conclusion} concludes.

\section{EIS and Consumption: A Two-Period Example}
\label{sec: example}

To exemplify and build intuition for our definition of EIS and main comparative statics results, we begin with a simple two period example without uncertainty. Time is indexed by $t\in\{1,2\}$. The agent is endowed with $e_t\in\R_{+}$ units of the consumption good in each period and can freely buy and sell a risk-free asset in period 1 with gross return $R_f\in\R_{++}$. Thus, the agent faces the lifetime budget constraint
\begin{equation*}
    c_1+\frac{c_2}{R_f}\le e_1+\frac{e_2}{R_f}.
\end{equation*}

In period 2, if the agent consumes $c_2\in\R_{+}$, their utility is simply $u_2(c_2)=c_2$. In period 1, the agent has recursive preferences over period 1 consumption $c\in\R_{+}$ and period 2 utility $v\in\R_{+}$ represented by the aggregator $f(c,v)$. To avoid corner or multiple solutions we assume that $f$ is twice continuously differentiable, strictly increasing in each argument, strictly quasi-concave, and satisfies the Inada condition.

In this setting, we define the EIS of $f$ as\footnote{The numerical value of EIS is invariant to a monotonic transformation of the utility function (aggregator). To see this, let $g(c,v)=F(f(c,v))$, where $F$ is strictly increasing and differentiable.  Then by the chain rule we have $g_c=F'f_c$ and $g_v=F'f_v$, so $g_c/g_v=f_c/f_v$.}
\begin{equation}
\psi=-\frac{\dd \log\left(\frac{c}{v}\right)}{\dd \log\left(\frac{f_c}{f_v}\right)},\label{eq:EIS_defex}
\end{equation}
where $f_x=\frac{\partial f}{\partial x}$ is the date 1 marginal utility of $x\in\{c,v\}$.\footnote{Since $\log (f_c/f_v)$ is a function, not a variable, the notation \eqref{eq:EIS_defex} is not rigorous.  Formally, given an arbitrary point $(c,v)$, let $s = \log (f_c / f_v)$, and define $c(s)$ and $v(s)$ that jointly solve 
$f(c(s),v(s))=f(c,v)$. Then, the EIS at a particular point is $\psi=-\frac{\dd \log(c/v)}{\dd s}$.} The EIS captures the substitutability between contemporaneous consumption and future utility. Indeed, under our assumption that the agent can trade a risk-free asset, the agent's first-order condition for optimal consumption implies that $\frac{f_c}{f_v}=R_f$. Thus, as $v=c_2$, the EIS is equivalent to the commonplace definition of the EIS as the elasticity of consumption growth to changes in the risk-free rate:
\begin{equation}
\label{eq:cg}
    \psi=-\frac{\dd \log\left(\frac{c_1}{c_2}\right)}{\dd \log R_f}.
\end{equation}

Critically, our definition of the EIS \eqref{eq:EIS_defex} depends solely on the agent's preferences, while \eqref{eq:cg} requires restrictions on investment opportunities and preferences that might be violated in practice:
\begin{enumerate*}
\item preferences could be nonhomothetic, \textit{i.e.,} utility might not be a linear function of period 2 consumption or
\item the household may face borrowing constraints.
\end{enumerate*}

Our main theoretical results characterize the consumption response to arbitrary shocks in terms of the agent's EIS and the wealth effects induced by these shocks. To parameterize such shocks in this example environment, suppose that the agent faces a preference shifter $\rho\in\R_+$ such that the agent's time 1 utility is given by $f(c,\rho v)$. The following proposition characterizes how consumption in period 1 is affected by changes in
\begin{enumerate*}
\item the risk-free rate and
\item the value of continuation utility.
\end{enumerate*}

\begin{prop}\label{prop:main_two}
The sign of the consumption response to a change in continuation value (at $\rho=1$) is given by
    \begin{equation}
    \label{eq:main_two_scaled}
       \sgn\left(\frac{\partial c}{\partial \rho}\right) = \sgn(1-\psi).
    \end{equation}
Moreover, the sign of the consumption response to a change in the risk-free rate is given by
\begin{equation}
\sgn\left(\frac{\partial c}{\partial R_f}\right)=\sgn(1-\varepsilon\psi),\label{eq:main_two}
\end{equation}
where $\varepsilon=\frac{e_1-c+e_2/R_f}{e_1-c} \ge 1$.
\end{prop}
\begin{proof}
See Appendix \ref{prop:main_twoproof}.
\end{proof}

Thus, the effect on consumption of a change in continuation values is exactly signed by the relationship of the EIS with unity. In particular, if and only if $\psi > 1$, when continuation values increase, consumption today decreases. Intuitively, when $\psi>1$, the agent is willing to substitute consumption today for the now relatively more valuable consumption in the future. However, were $\psi<1$, the agent would increase consumption today because consumption today and consumption tomorrow are sufficiently complementary.

The response of consumption to changes in the interest rate depends on both the EIS and the wealth effects that the change in interest rates induces. These wealth effects are summarized by the ratio of lifetime wealth to contemporaneous wealth $\varepsilon$. When the agent receives no additional wealth in the future, $e_2=0$, this wealth effect is neutral, $\varepsilon=1$, and the consumption response to an interest rate shock is signed by the relationship of the EIS with unity. However, when the agent receives wealth in the future, an increase in $R_f$ reduces the value of the agent's endowment in period 1 as borrowing forward that wealth to period 1 is more expensive. This causes the agent to experience a negative wealth effect in period 1, which makes the agent more predisposed to cut period 1 consumption. As a result, consumption now falls under the less stringent condition that $ \psi > \frac{1}{\varepsilon}$. In our general model, $\varepsilon$ is the REMV, which we previewed in the introduction and will shortly define formally.\footnote{For a proportional shock to continuation values, the REMV equals one, which is why $\varepsilon$ does not \textit{explicitly} appear in \eqref{eq:main_two_scaled}.}
   
As a concrete example of the above, consider the Epstein-Zin (or CES) aggregator
\begin{equation*}
f(c,\rho v)=\left((1-\beta)c^{1-1/\psi}+\beta (\rho v)^{1-1/\psi}\right)^\frac{1}{1-1/\psi},\label{eq:ez_aggregator}
\end{equation*}
where $0<\beta<1$ is the discount factor. Using calculus (see Proposition \ref{prop:ezsol} in Appendix \ref{lem:homoproof} for a general solution in a dynamic environment), we  obtain
\begin{equation*}
(c,\rho v)=\left(\frac{(1-\beta)^\psi(e_1+e_2/R_f)}{(1-\beta)^\psi +\beta^\psi(\rho R_f)^{\psi-1}},\frac{\beta^\psi (\rho R_f)^\psi (e_1+e_2/R_f)}{(1-\beta)^\psi+\beta^\psi (\rho R_f)^{\psi-1}}\right).\label{eq:ez_demand}
\end{equation*}
Figure \ref{fig:CEIS_ics} plots the budget sets, indifference curves, and optimal consumption bundles for different values of the interest rate (0\%, 25\%, and 50\%) when $e_2=0$ and $\beta = 1/2$. Different columns correspond with different choices of EIS, $\psi=1/2,1,2$.

\begin{figure}[t]
\centering
\includegraphics[width = 0.95\linewidth]{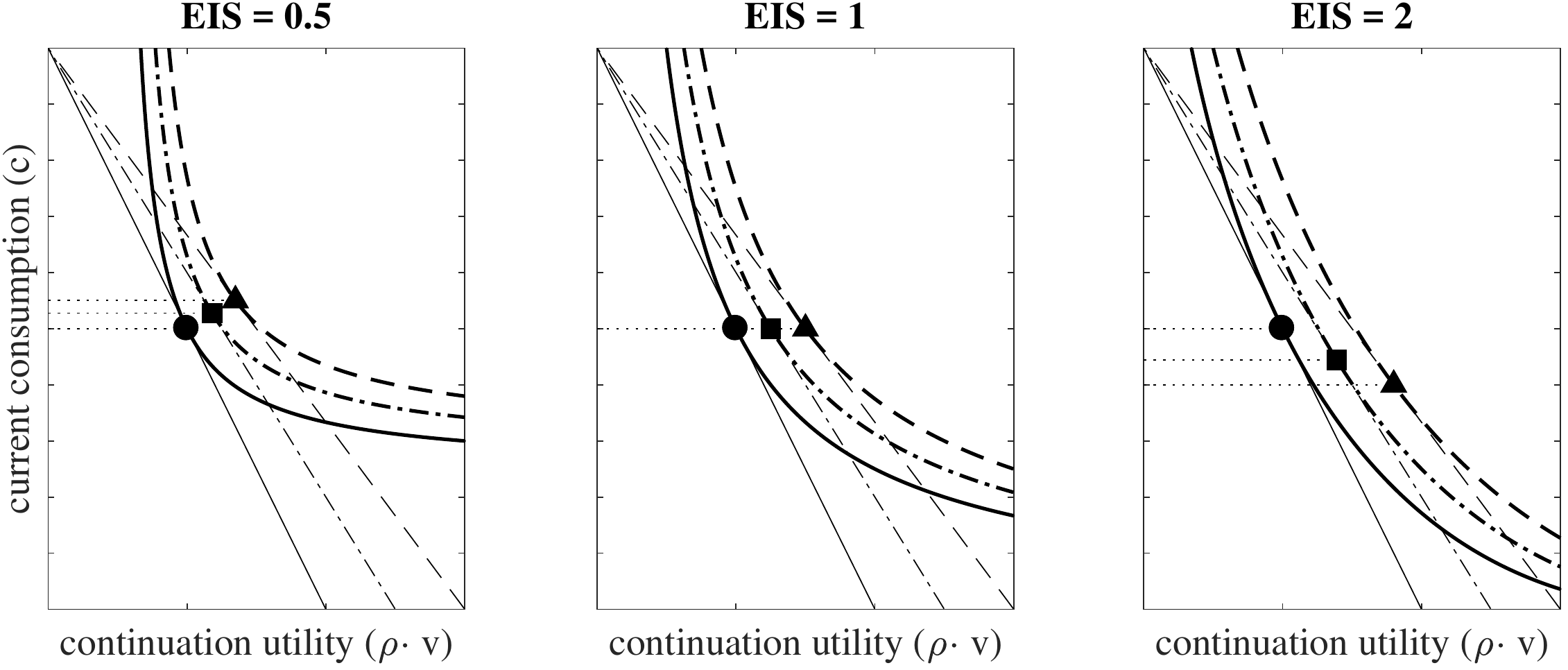}
\caption{Response of consumption to $R_f$ (or $\rho$) with Epstein-Zin preferences.}\label{fig:CEIS_ics}
\end{figure}

In the middle panel, EIS equals one, and the aggregator is Cobb-Douglas. Hence, the agent spends a constant fraction of wealth on each good, so consumption is invariant to $R_f$. When EIS is less than unity (left panel), consumption and continuation utility are gross complements, so the agent consumes more of both goods. The opposite is the case in the right panel, in which the two goods are gross substitutes as the EIS exceeds one.

Our main theoretical results generalize, to dynamic and stochastic environments with much less structure, this basic insight that the consumption response to \textit{any} shock to continuation values is characterized by the relationship between the product of the EIS and wealth effects ($\varepsilon$) with one. This allows us to derive comparative statics in applications and provide insight into strategies by which the EIS might be robustly estimated.  

\section{Model and Main Results}
\label{sec:results}
We now introduce our general framework and derive our main results relating the EIS with consumption responses to shocks.

\subsection{Model Primitives}
Time is discrete, finite, and indexed by $t\in\mathcal{T}=\{0,1,\ldots,T\}$. All random variables are defined with respect to a probability space $(\Omega,\mathcal{F},P)$. A single agent has preferences defined over random consumption plans $\{c_{t+s}\}_{s=0}^{T-t}$ for all $t\in\mathcal{T}$ that are constructed recursively as follows. Terminal utility is $U_T=u_T(c_T)$ for some $u_T:\R_{+}\to\R_{+}$. Given a random continuation utility $U_{t+1}\in\mathcal{U}$, the time $t$ recursive utility is given by
\begin{equation*}
    U_t = f_t\left(c_t,\mathcal{M}_t\left(U_{t+1}\right)\right),
\end{equation*}
where $f_t:\R_{+}\times\R_{+}\to\R_{+}$ is upper semi-continuous and aggregates consumption and a certainty equivalence functional of the distribution of the continuation value $\mathcal{M}_t:\mathcal{U}\to\R_{+}$.

Suppose that the agent has financial wealth $w_t\in\R_{++}$ at time $t$. The agent can invest this wealth in portfolios $\theta_t\in\Theta_t$. When the agent invests $w-c$ in portfolio $\theta_t$, their continuation wealth next period is the random variable $W_{t+1}(w-c,\theta_t)$. Define the value function at date $t$ given wealth $w$ as $V_t(w)$. By the principle of optimality and backward induction, the value function is given by the Bellman equation
\begin{equation*}
    V_t(w)=\sup_{c\in[0,w],\theta_t\in\Theta_t} f_t\left(c,\mathcal{M}_t\left(V_{t+1}\left(W_{t+1}(w-c,\theta_t)\right)\right)\right).
\end{equation*}
To index our comparative statics, we parameterize the aggregation functional, the continuation wealth function, and the subjective distribution over continuation valuations by a scalar parameter $\alpha\in[0,1]$ and define the \textit{continuation value function} $v_t:\R_{+}\times[0,1]\to\R_{+}$
\begin{equation}
\label{eq:contdef}
    v_t(w,\alpha)=\sup_{\theta_t\in\Theta_t^{\alpha}}\mathcal{M}^{\alpha}_t\left(V^{\alpha}_{t+1}(W^{\alpha}_{t+1}(w,\theta_t))\right)
\end{equation}
with the normalization that $v_t(w,\cdot)$ is an increasing function for all $w\in\R_{+}$. Suppressing time subscripts, we can express the consumption-savings problem of the agent as
\begin{equation}
\label{eq:static-formulation}
    \sup_{c\in[0,w]} f(c,v(w-c,\alpha))
\end{equation}
and define the pair $E=(f,v)$ as an environment.

\begin{remtri}[Dynamic Consistency and Infinite Horizon are Inessential] 
All of our theoretical results hold even if the agent does not expect to optimize from date $t+1$ onward. That is, they may be time-inconsistent or even not necessarily in control of any future decisions. In this setting, we can still define $v_t$ as in \eqref{eq:contdef} with $V_{t+1}$ replaced by $U_{t+1}$ and our analysis follows as written. This also makes clear that the agent's problem could have infinite horizon. So long as the agent's value function exists, we can still study the consumption-savings decision of the agent by studying \eqref{eq:static-formulation}.
\end{remtri}

Toward understanding the consumption response to shocks using local perturbations, we define environments in which this approach is generally possible as (strongly) regular:

\begin{defn}
\label{def:reg}
The environment $E=(f,v)$ is \emph{regular} if:
\begin{enumerate}
    \item The aggregator $f$ is strictly increasing and twice continuously differentiable with positive cross-partial derivative.
    \item The continuation value function $v$ is strictly increasing and twice continuously differentiable.
\end{enumerate}
If, in addition, the following is satisfied, then the environment is \emph{strongly regular}:
\begin{enumerate}
\setcounter{enumi}{2}
    \item All solutions to \eqref{eq:static-formulation}, $c:\R_{+}\times[0,1]\to\R_{+}$, are such that $c(w,\alpha)\in (0,w)$.
\end{enumerate}
\end{defn}

We provide sufficient conditions in terms of primitives $\{u_T,(\Omega,\mathcal{F},P),\{f_t,\mathcal{M}_t,\Theta_t,W_t\}_{t\in\mathcal{T}}\}$ such that the induced environments $\{(f_t,v_t)\}_{t\in\mathcal{T}}$ are strongly regular for almost all levels of wealth in Lemma \ref{lem:reg} in Appendix \ref{ap:suffcond}.

\subsection{Main Results: The EIS and Consumption Responses to Shocks}
Toward characterizing the consumption response to shocks, we first define both the concepts of the EIS and REMV. The EIS is the elasticity of substitution between current consumption and future continuation value:

\begin{defn}[Elasticity of Intertemporal Substitution]
The \emph{EIS} is
\begin{equation}
\label{eq:EIS_def}
\psi=-\frac{\frac{\partial \log\left(\frac{c}{v}\right)}{\partial \alpha}}{\frac{\partial \log\left(\frac{f_c}{f_v}\right)}{\partial\alpha}},
\end{equation}
where all partial derivatives are evaluated at $c$.
\end{defn}

The REMV is the ratio between the elasticity of the marginal value of wealth with respect to the shock and the elasticity of the value of wealth with respect to the shock:

\begin{defn}[Relative Elasticity of the Marginal Value of Wealth]
The \emph{REMV} is
\begin{equation}\label{eq:epsilon}
\varepsilon=\frac{\frac{\partial \log v_w}{\partial\alpha}}{\frac{\partial \log v}{\partial\alpha}},
\end{equation}
where all partial derivatives are evaluated at $c$.
\end{defn}

Intuitively, the REMV measures the impact of any wealth effects in the consumption response to shocks while the EIS measures substitution effects.

In strongly regular environments, the following result establishes a formula for the consumption response to changes in the continuation value. It moreover shows, under the benchmark condition that the continuation value of wealth is concave, that the sign of the consumption response to a positive shock to continuation value is characterized by the relationship of the product of the EIS and REMV with unity.

\begin{thm}\label{thm1}
If the environment $(f,v)$ is strongly regular and $v_{ww}\le 0$, then
\begin{equation}
\label{eq:thm1}
    \sgn\left(\frac{\partial c}{\partial \alpha}\right) = \sgn(1-\varepsilon\psi).
\end{equation}
\end{thm}
\begin{proof}
See Appendix \ref{thm1proof}.
\end{proof}

We prove this result by applying the implicit function theorem to the necessary first-order condition for optimal consumption and re-expressing the resulting equation in terms of the EIS and REMV. This yields the following formula for the consumption response to shocks (which holds even when continuation value functions are not concave in wealth):
\begin{equation}
\left(\frac{1}{c}+\frac{v_w}{v}-\psi\frac{v_{ww}}{v_w}\right)c_\alpha=\frac{v_\alpha}{v}(1-\varepsilon\psi).\label{eq:main_gen}
\end{equation}
It follows immediately from \eqref{eq:main_gen} that when the continuation value function is concave in wealth, the sign of the consumption response is signed by $1-\text{REMV}\times\text{EIS}$.

To understand the intuition behind this result, note that consumption increases in response to increased continuation value if and only if $\varepsilon\psi\le1$. When wealth effects are neutral (\textit{i.e.,} $\varepsilon=1$), this reduces to the familiar condition that $\psi\le 1$ that simply asks if consumption today and tomorrow are gross complements. If they are gross complements, then the gain in continuation utility from an increase in $\alpha$ induces additional consumption today as the agent wishes to increase consumption today and utility tomorrow in tandem.

However, in general, wealth effects through the REMV complicate this relationship. If the marginal value of wealth increases proportionally more in response to the shock to continuation values than the value of wealth, then $\varepsilon>1$. In this case, the relative rise in the marginal value of wealth makes saving more attractive. As a result, consumption today and continuation value must now be sufficiently complementary to overcome this wealth effect and consumption only increases today if $\psi\le \frac{1}{\varepsilon}<1$.

For a concrete example, consider the effects of changes on interest rates on consumption in our simple two-period example. Observe that
Proposition \ref{prop:main_two} is a special case of Theorem \ref{thm1} that sets $v(w)=R_fw+e_2$ and $\alpha=R_f$. Since \eqref{eq:main_two} and \eqref{eq:thm1} are consistent with one another, it follows that $\varepsilon = \frac{e_1-c+e_2/R_f}{e_1-c}$ in the two period example coincides precisely with our general definition of the REMV. Intuitively, when $e_2 > 0$ and the consumer is a saver $(e_1 - c > 0)$, we have $\varepsilon > 1$ and wealth effects induced by changes in $R_f$ are smaller relative to the case with $e_2 = 0$. This follows because the decrease in the present value of the endowment partially offsets the benefits associated with saving at the higher interest rate.

\begin{remtri}[Consumption Responses to Discrete Changes in $\alpha$]
\label{rem:discrete}
While expressed locally, Theorem \ref{thm1} can be used to provide robust comparative statics for consumption responses to discrete changes in $\alpha$. Concretely, suppose we want to know the consumption response to a change in $\alpha$ from $\alpha_0$ to $\alpha_1$. If the environment is strongly regular, we have that
\begin{equation*}
c(w,\alpha_1)-c(w,\alpha_0)=\int_{\alpha_0}^{\alpha_1}c_{\alpha}(w,s)\dd s.
\end{equation*}
Thus, by substituting \eqref{eq:main_gen} into the integral, we have a formula for the discrete change. Most importantly, when the continuation value function is concave in wealth ($v_{ww}\le 0$), if we know the sign of the function $1-\varepsilon(w,\alpha)\psi(w,\alpha)\gtrless0$ for all $\alpha\in[\alpha_0,\alpha_1]$, then we know that $\int_{\alpha_0}^{\alpha_1}c_{\alpha}(w,s)\dd s \gtrless0$, and therefore that $c(w,\alpha_1)-c(w,\alpha_0)\gtrless 0$. Moreover, even changes in objects with no obvious continuous counterpart can be parameterized in a smooth way by $\alpha$. Concretely, suppose that we want to understand the consumption effect of transitioning from $\alpha_0$ to $\alpha_1$, we can always parameterize as
\begin{equation*}
\label{eq:ctspar}
    v(w,\alpha)=\frac{\alpha-\alpha_0}{\alpha_1-\alpha_0} v(w,\alpha_1)+\left(1-\frac{\alpha-\alpha_0}{\alpha_1-\alpha_0}\right)v(w,\alpha_0)
\end{equation*}
and apply the above formulas. Thus, the \textit{global} relationship of the product of the EIS and REMV with unity is sufficient to sign the consumption response to discrete shocks. 
\end{remtri}

More broadly, even in cases where optimal consumption is not interior, while we cannot characterize the response of consumption, we can still provide sufficient conditions for globally monotone responses of consumption to $\alpha$ using techniques from the literature on monotone comparative statics \citep{MilgromShannon1994}. Indeed, under only the hypotheses that $f$ is twice continuously differentiable, $v$ is continuously differentiable, and $v_{w}$ is continuously differentiable in $\alpha$, we obtain the following sufficient condition for every possible optimal consumption function to be increasing in $\alpha$.

\begin{thm}
\label{thm2}
Suppose the environment $(f,v)$ is regular. If
\begin{equation}
\label{eq:mcscond}
  \frac{\frac{v_w}{v}}{\frac{f_c}{f}}\left(\frac{\frac{v_{w\alpha}}{v_w}}{\frac{v_{\alpha}}{v}}+\frac{vf_{vv}}{f_v}\right)\left(\frac{f_{cv}f}{f_cf_v}\right)^{-1}<1,
\end{equation}
then any optimal $c$ is increasing in $\alpha$. Under the reverse inequality, any optimal $c$ is decreasing in $\alpha$.
\end{thm}
\begin{proof}
See Appendix \ref{thm2proof}.
\end{proof}

To understand condition \eqref{eq:mcscond}, observe that consumption is increasing in $\alpha$ so long as the LHS is bounded above by 1. The LHS is the product of three terms. The first is the marginal value of wealth in units of the marginal value of consumption, which indexes the value of wealth effects in consumption equivalent units. The second is the sum of the REMV (now extended away from the optimum) and the curvature of the aggregator in continuation value, which together index the size of wealth effects. Thus, the first two terms represent the total wealth effect from global shocks to $\alpha$. 

The third and final term is the inverse of the normalized complementarity of consumption and continuation value for the aggregator, which indexes the substitutability of consumption and continuation values. This mimics the role of the EIS but does so globally instead of just around the optimum. To see this, observe that when $f(c,v)$ is of the Epstein-Zin form in which the EIS is globally constant and equal to $\psi$, then
\begin{equation*}
f(c, v)=\left((1-\beta)c^{1-1/\psi}+\beta v^{1-1/\psi}\right)^\frac{1}{1-1/\psi} \implies     \left(\frac{f_{cv}(c,v)f(c,v)}{f_c(c,v)f_v(c,v)}\right)^{-1} \equiv \psi.
\end{equation*}
Hence, even away from the optimum where the EIS is defined, Theorem \ref{thm2} provides an analogous condition to that provided by Theorem \ref{thm1}: the product of wealth effects (represented by an extended REMV) and substitution effects (represented by an extended EIS) must be less than unity for consumption to increase in response to a positive continuation value shock.

\subsection{Consumption Responses under Homotheticity}
So far we have seen how the consumption response to a shock to continuation values depends on both the EIS and REMV. Therefore, to isolate the role of the EIS, it is illustrative to consider environments in which the REMV is known to equal one. The most natural and commonly occurring such environments are ones which are \textit{homothetic}.

\begin{defn}
\label{def:hom}
An environment $(f,v)$ is \emph{homothetic} if $f$ is homogeneous of degree one and $v$ is a strictly increasing, linear function, \textit{i.e.,} $v(w,\alpha)=g(\alpha) w$ with $g(\alpha)>0$.
\end{defn}

Under homotheticity, we know both that
\begin{enumerate*}
\item the continuation value function is concave in wealth as it is linear ($v_{ww}=0$) and
\item the REMV is identically equal to unity ($\varepsilon=1$).
\end{enumerate*}
Thus, in homothetic environments, the following Corollary of Theorem \ref{thm1} characterizes the sign of the consumption response to continuation value shocks in terms of the relationship of the EIS to unity.

\begin{cor}
\label{cor:hom}
If the environment $(f,v)$ is homothetic and strongly regular, then
\begin{equation}
\label{eq:cor1}
    \sgn\left(\frac{\partial c}{\partial \alpha}\right) = \sgn(1-\psi)
\end{equation}
\end{cor}
\begin{proof}
See Appendix \ref{cor:homproof}.
\end{proof}

The intuition for this result is, of course, that homotheticity makes wealth effects neutral. To see this algebraically, observe that we can write the REMV under homotheticity as
\begin{equation}
\label{eq:epshom}
    \varepsilon=\frac{\frac{v_{w\alpha}}{v_w}}{\frac{v_{\alpha}}{v}} = \frac{\frac{g'(\alpha)}{g(\alpha)}}{\frac{g'(\alpha)w}{g(\alpha)w}} = 1.
\end{equation}

In homothetic environments, we can similarly simplify the sufficient condition for the consumption function to be monotone in $\alpha$ from Theorem \ref{thm2}:

\begin{cor}
\label{cor:hommcs}
Suppose the environment $(f,v)$ is homothetic and regular. If
\begin{equation*}
 \frac{\frac{1}{w}}{\frac{f_c}{f}}\left(\frac{f_{cv}f}{f_cf_v}\right)^{-1} <1,
\end{equation*}
then any optimal $c$ is increasing in $\alpha$. Under the reverse inequality, any optimal $c$ is decreasing in $\alpha$.
\end{cor}
\begin{proof}
See Appendix \ref{cor:hommcsproof}.
\end{proof}

This result exploits the neutrality of wealth effects and homogeneity of the aggregator. Thus, under homotheticity, so long as the product of the marginal value of wealth in consumption units and the extended notion of the EIS is less than one, any optimal consumption function is increasing. For practical purposes, this can be checked by verifying the simpler condition that $ \frac{1}{w}\frac{f_v}{f_{cv}}<1$.

\section{Applications}
\label{sec:applications}
In this section, we apply our comparative statics results to study how consumption responds to various shocks in three applications to portfolio allocation problems, consumption-savings problems, and entrepreneurial investment problems.

\subsection{Risk, Ambiguity, and Investment}
An investor is endowed with initial wealth $w_0$ and receives no additional wealth in future periods. Each period, they decide how to optimally invest their wealth $w_t$. The set of feasible portfolios at each time $t$ is given by $\Theta_t$. For each $\theta\in\Theta_t$, the random variable $R_{t+1}(\theta)$ yields the gross return on invested wealth today. Thus, the total return on invested wealth is given by $w_{t+1}=R_{t+1}(\theta_t)(w_t-c_t)$. The agent's intertemporal aggregators are given by $f_t$. We now consider how shocks in this setting affect consumption and investment in environments with certainty equivalence functionals that feature the potential for risk aversion and ambiguity aversion.

\subsubsection{Risk}
\label{sec:risk}
To model risk aversion, we suppose that the certainty equivalence functional is given by the standard quasi-arithmetic form,
\begin{equation}
\label{eq:defM}
    \mathcal{M}_t(U)=\phi_t^{-1}\left(\E_t\left[\phi_t\left(U\right)\right]\right),
\end{equation}
where $\phi_t:\R_{+}\to\R$ is a strictly increasing and concave function.

We study the following five comparative statics to the agent's preferences and investment opportunities in this setting:
\begin{enumerate}
    \item\label{item:cs1} Risk aversion increases: $\phi_t$ changes to $\tilde{\phi}_t$, where $g=\tilde{\phi}_t\circ \phi_t^{-1}$ is increasing and concave.
    \item\label{item:cs2} The investment opportunity set shrinks: $\Theta_t$ changes to $\tilde{\Theta}_t\subset \Theta_t$.
    \item\label{item:cs3} The portfolio returns become lower:
    $\{R_{t+1}(\theta)\}_{\theta\in\Theta_{t}}$ changes to $\{\tilde R_{t+1}(\theta)\}_{\theta\in\Theta_{t}}$ where $R_{t+1}(\theta)\succeq_\mathrm{FOSD}\tilde R_{t+1}(\theta)$ for all $\theta\in\Theta_t$
    \item\label{item:cs4} The portfolio returns become riskier: $\{R_{t+1}(\theta)\}_{\theta\in\Theta_{t}}$ changes to $\{\tilde R_{t+1}(\theta)\}_{\theta\in\Theta_{t}}$ where $R_{t+1}(\theta)\succeq_\mathrm{SOSD}\tilde R_{t+1}(\theta)$ for all $\theta\in\Theta_t$
    \item\label{item:cs5} Future consumption expenditure becomes less valuable: the aggregator shifts from $f_{t+1}$ to $\tilde{f}_{t+1}$ where $\tilde{f}_{t+1}(c,v) = f_{t+1}(c/g(c,v),v)$, where $g(c,v) \ge 1$ for all $(c,v)$.
\end{enumerate}

\begin{prop}
\label{prop:risk}
Comparative statics \ref{item:cs1}--\ref{item:cs3} and \ref{item:cs5} lower the agent's continuation value function at all previous dates. If the agent's value functions are concave, then comparative static \ref{item:cs4} lowers the agent's continuation value function at all previous dates.
\end{prop}
\begin{proof}
See Appendix \ref{prop:riskproof}.
\end{proof}

Thus, we can parameterize \ref{item:cs1}--\ref{item:cs5} by an arbitrary smooth transformation indexed by $\alpha$.\footnote{Recall by Remark \ref{rem:discrete} how this is possible even for discrete changes, such as changes in the investment opportunity set.} It follows that when the environment is strongly regular, Theorem \ref{thm1} immediately implies that consumption decreases in response to any of these changes if and only if $\varepsilon\psi\le1$.

As before, these comparative statics are complicated by the presence of wealth effects through the REMV. However, under the following benchmark assumptions, the environment is homothetic and $\varepsilon=1$.

\begin{lem}
\label{lem:homo}
The induced environments $\{(f_t,v_t)\}_{t\in\mathcal{T}}$ are homothetic if
\begin{enumerate}
    \item the aggregators $\{f_t\}_{t\in\mathcal{T}}$ are weakly increasing, strictly quasi-concave, and homogeneous of degree one,
    \item the certainty equivalent is of the CRRA form: $\phi_t(x)=x^{1-\gamma_t}$ for $\gamma_t\in\R_{+}/ \{1\}$ and $\phi_t(x)=\log x$ for $\gamma_t=1$,
    \item the sets of potential portfolios $\{\Theta_t\}_{t\in\mathcal{T}}$ are compact and returns $\{R_{t+1}(\theta)\}_{\theta\in\Theta_t,t\in\mathcal{T}}$ are bounded, and
    \item the terminal utility is proportional to consumption $u_T(c)=b_Tc$ for some random variable $b_T>0$.
\end{enumerate}
\end{lem}
\begin{proof}
See Appendix \ref{lem:homoproof}.
\end{proof}

Thus, under these conditions, increases in risk aversion, decreased investment opportunities, lower and riskier returns, and lower future value to consumption all decrease consumption if and only if $\psi\le 1$. As a concrete illustration, in Appendix \ref{lem:homoproof}, we leverage this result to provide an explicit solution for optimal consumption in the case where the agent has Epstein-Zin preferences (Proposition \ref{prop:ezsol}).

Moreover, in Appendix \ref{lem:homoproof}, we extend the environment in this section to allow for stochastic death and (linear) bequest motives and show that Lemma \ref{lem:homo} continues to hold (Remark \ref{rem:bequest}). Moreover, adverse shocks to the value of bequests reduce continuation value functions. Thus, the response of the agent to increased estate taxes is to decrease contemporaneous consumption if and only if $\psi\le1$. 

\subsubsection{Ambiguity}
We now study the consumption response to ambiguity in situations where a decision maker considers multiple prior distributions over the state. To model such situations where the decision maker is ambiguity-averse, we follow the approach to modelling ambiguity of \cite{hayashi-miao2011} and \cite{ju-miao2012} and consider certainty equivalents of the form
\begin{equation}
\mathcal{M}_t(U)=\varphi_t^{-1}(\E_{\mu_t}[\varphi_t(\phi_t^{-1}(\E_{\pi_t}[\phi_t(U)]))]),\label{eq:ambiguity}
\end{equation}
where $\phi_t$ and $\varphi_t$ capture risk aversion and ambiguity aversion, respectively.  Here $\pi_t\in \mathcal{P}_t$ is the subjective probability measure over the state space, and $\mu_t$ is the subjective probability measure over the set of the underlying stochastic process $\mathcal{P}_t$.  When $\varphi_t=\phi_t$, \eqref{eq:ambiguity} reduces to \eqref{eq:defM}, where the expectation is taken over $\mu_t\circ \pi_t$.  If the agent is infinitely ambiguity averse, then \eqref{eq:ambiguity} reduces to:
\begin{equation*}
\mathcal{M}_t(U)=\phi_t^{-1}\left(\min_{\pi_t\in \mathcal{P}_t}\E_{\pi_t}[\phi_t(U)]\right),\label{eq:multiprior}
\end{equation*}
which is the classical multi-priors model introduced by \cite{gilboa-schmeidler1989} and generalized to the intertemporal setting (without the separation of EIS from risk aversion) by \cite{epstein-schneider2003} and (with the three-way separation between EIS, risk aversion, and ambiguity aversion) by \cite{hayashi2005}.

In these settings, we consider the following two comparative statics:
\begin{enumerate}
\setcounter{enumi}{5}
    \item\label{item:cs6} The agent becomes more ambiguity averse in the smooth environment: $\varphi_t$ changes to $\tilde{\varphi}_t$, where $g=\tilde{\varphi}_t\circ \varphi_t^{-1}$ is increasing and concave.
    \item\label{item:cs7} The agent considers more prior distributions in the infinitely ambiguity averse environment: $\mathcal{P}_t$ changes to $\tilde{\mathcal{P}}_t$ with $\tilde{\mathcal{P}}_t\supset \mathcal{P}_t$.
\end{enumerate}

\begin{prop}\label{prop:ambiguity}
Comparative statics \ref{item:cs6} and \ref{item:cs7} lower the agent's continuation value function at all previous dates.
\end{prop}
\begin{proof}
See Appendix \ref{prop:ambiguityproof}.
\end{proof}

Thus, once again, we can parameterize these changes by a smooth transformation indexed by $\alpha$. When the environment is strongly regular, Theorem \ref{thm1} again immediately implies that consumption decreases in response to any of these changes if and only if $\varepsilon\psi\le1$.

\subsection{Consumption-Savings Problems with Income Risk}
An agent facing a stochastic income stream and a borrowing constraint decides how to optimally save. The agent's income in each period $y_t$ is the product of a permanent income component $p_t$ and an IID transitory income shock $\tau_t$. Moreover, permanent income evolves according to a geometric random walk with IID shocks $\eta_t$. The household can save its wealth in a variety of portfolios $\theta\in\Theta_t$ that yield random gross returns $R_{t+1}(\theta)$, which potentially allow for both hedging labor income risk and investing in financial markets. The agent faces a borrowing constraint such that $c_t\le w_t$. 

The agent has quasi-arithmetic risk preferences with function $\phi_t$ and aggregates consumption and continuation values according to $f_t$. The household's terminal utility function is linear in consumption $u_T(c_T)=b_Tc_T$ for $b_T>0$.

In the terminal period, the household's value function is $V_T(w_T,p_T)=b_T w_T$. In all previous periods, the value function is defined recursively by
\begin{equation*}
V_t(w_t,p_t)=\max_{c_t,w_{t+1},\theta_t\in\Theta_t}f_t\left(c_t,\phi_t^{-1}\left(\E_t\left[\phi_t\left(V_{t+1}(w_{t+1},p_{t+1})\right)\right]\right)\right),
\end{equation*}
where
\begin{equation*}
\begin{split}
w_{t+1}&=R_{t+1}(\theta_t)(w_t-c_t)+y_{t+1},\\
y_{t} &= p_t\tau_{t},\\
p_t&=p_{t-1}\eta_t,\\
c_t&\le w_t.
\end{split}
\end{equation*}

We study the following three comparative statics in this setting:

\begin{enumerate}
\setcounter{enumi}{7}
    \item\label{item:cs8} The agent's income permanent or transitory income falls: the distribution of $\tau_t$, denoted by $F_{\tau_t}$, or $\eta_t$, denoted by $F_{\eta_t}$, becomes $\tilde{F}_{\tau_t}$ or $\tilde{F}_{\eta_t}$ with $F_{\tau_t}\succeq_\mathrm{FOSD} \tilde F_{\tau_t}$ or $F_{\eta_t}\succeq_\mathrm{FOSD} \tilde F_{\eta_t}$.
    
    \item\label{item:cs9} The household's investment or hedging opportunities shrink: $\Theta_t$ changes to $\tilde{\Theta}_t\subset \Theta_t$.
    \item\label{item:cs10} The agent's income becomes riskier: the distribution of $\tau_t$, denoted by $F_{\tau_t}$, or $\eta_t$, denoted by $F_{\eta_t}$, becomes $\tilde{F}_{\tau_t}$ or $\tilde{F}_{\eta_t}$ with $F_{\tau_t}\succeq_\mathrm{SOSD} \tilde F_{\tau_t}$ or $F_{\eta_t}\succeq_\mathrm{SOSD} \tilde F_{\eta_t}$.
\end{enumerate}

\begin{prop}
\label{prop:cs}
Comparative statics \ref{item:cs8} and \ref{item:cs9} lower the agent's continuation value function at all previous dates. When the agent's value functions are concave, comparative static \ref{item:cs10} lowers the agent's continuation value function at all previous dates.
\end{prop}
\begin{proof}
See Appendix \ref{prop:csproof}.
\end{proof}

Thus, once more, we can parameterize these changes by a smooth transformation indexed by $\alpha$. When the environment is strongly regular, Theorem \ref{thm1} again immediately implies that consumption decreases in response to any of these changes if and only if $\varepsilon\psi\le1$.

Of course, the REMV $\varepsilon$ complicates the relationship between consumption responses to shocks and the EIS. For example, a reduction in hedging opportunities can change the marginal values of permanent income and wealth through a precautionary savings channel. Nevertheless, under benchmark assumptions, we can derive an exact formula for the REMV.

\begin{lem}
\label{lem:remv}
Suppose the agent's aggregators $\{f_t\}_{t\in\mathcal{T}}$ are homogeneous of degree one and their certainty equivalence functionals $\{\mathcal{M}_t\}_{t\in\mathcal{T}}$ are of the CRRA form. The REMV is given by
\begin{equation}
\label{eq:remveq}
    \varepsilon = \frac{1+\frac{p}{w}\frac{v_p}{v_w}}{1+\frac{p}{w}\frac{v_{p\alpha}}{v_{w\alpha}}}.
\end{equation}
\end{lem}
\begin{proof}
See Appendix \ref{lem:remvproof}.
\end{proof}

This expression makes clear that while the REMV complicates the relationship between consumption and the EIS, it does so precisely to the extent that permanent labor income is relatively important to the household when compared to financial wealth. Indeed, for households who have little permanent income relative to wealth ($p/w\approx 0$), \eqref{eq:remveq} implies $\varepsilon\approx1$, and the consumption response is characterized by the relationship of the EIS with unity.  Indeed, structural estimates from \cite{GourinchasParker2002} suggest that the behavior of older, high net worth households (a group that owns a large share of total financial wealth overall) is primarily driven by life-cycle (bequest/retirement) motives, rather than these precautionary concerns.

\subsection{Entrepreneurial Investment}
Understanding the consumption-savings behavior of entrepreneurs is critical to asset pricing as business owners are vastly over-represented at the top of the wealth distribution \citep{smith2019capitalists}. Moreover, unlike households in models with exogenous labor income, business owners' earnings can plausibly respond to changes in investment opportunities, a feature which can more easily preserve homotheticity of the problem and thus avoid some of the challenges posed by the REMV. To illustrate this, this section considers an environment with an entrepreneur who decides how much to produce, consume and invest in both her own capital stock and financial markets. Concretely, the agent has two investment opportunities:
\begin{enumerate*}
\item investing $a_t$ dollars in portfolios $\theta\in\Theta_t$ of risky financial assets with (random) gross return $R_{t+1}(\theta)=1+r_{t+1}(\theta)$ and
\item purchasing capital $k_t$ at price $P_{kt}$, which is subject to productivity ($z_t$) and stochastic depreciation ($\delta_t$) shocks.
\end{enumerate*}

The entrepreneur has access to the production technology
\begin{equation*}
    y_t = z_t g_t(k_t,l_t),
\end{equation*}
where $g_t$ is homogeneous of degree one, $l_t$ is the number of efficiency units of labor hired in a competitive labor market at wage $\nu_{t}$ per efficiency unit, and $k_t$ is the firm's capital stock, which evolves according to
\begin{equation*}
    k_t = (1-\delta_t)k_{t-1} + i_t / P_{kt}.
\end{equation*}
Importantly, negative investment (\textit{i.e.}, liquidation of capital) is permitted.

Due to an unmodeled agency friction, the firm is only able to borrow  $b_t \in [0,\lambda (P_{kt} k_t)] $ one period debt at rate $1+r_{b,t+1}$ for some constant $\lambda \in (0,1)$ and is otherwise unable to raise external sources of financing.\footnote{We assume for simplicity that the distribution of $\delta_t$ and changes in capital prices ensure that $P_{k,t+1}(1-\delta_{t+1})- \lambda P_{kt}R_{b,t+1} >0$ and debt is default-free. Relaxing this comes at the expense of additional notation and assumptions about what happens in case of default, but simple extensions with defaultable debt preserve the homogeneity in net worth.} 
Stochastic depreciation shocks are uninsurable and hit before investment decisions are made, exposing the agent to idiosyncratic business-specific risk. Moreover, investment returns, profits, and changes in the value of the capital stock are assumed to be taxed at a common capital tax rate $\tau_t$.

Under these assumptions, the entrepreneur's net worth evolves according to
\begin{multline*}
w_{t+1} = a_t(1+(1-\tau_{t+1})r_{t+1}(\theta_t))+\{(1-\tau_{t+1})P_{k,t+1}(1-\delta_{t+1})+\tau_{t+1} P_{kt} \}k_t \\
+ (1-\tau_{t+1})[y_{t+1} - \nu_{t+1} l_{t+1}] -b_t(1+(1-\tau_{t+1})r_{b,t+1}], \label{eq:ent_bc_own}
\end{multline*}
where financial assets $a_t$ and $b_t$ are both weakly positive and we assume that all returns are bounded.
In addition, the entrepreneur faces the simple budget constraint
\begin{equation}
w_t- c_t  = P_{kt}k_t- b_t + a_t, \label{eq:ent_bc2} 
\end{equation}
so net worth, after consumption, is split between a portfolio of financial assets or the firm.\footnote{Alternatively, we could consider the simpler environment where the entrepreneur rents the capital on a period-by-period basis,
\begin{equation*}
w_{t+1}=a_t(1+(1-\tau_{t+1})r_{t+1}(\theta_t))+(1-\tau_{t+1})[y_{t+1} - \phi_{t+1} l_{t+1} - \phi_{kt+1} k_{t+1}],
\end{equation*}
where $\phi_{k,t}$ is the rental rate of capital, but the right hand side of the budget constraint is the same as \eqref{eq:ent_bc2} except that $k_t=b_t=0$. In this setting, increases in rental rates would also lower the agent's continuation value function.}

The entrepreneur has CRRA preferences over risk:
\begin{equation*}
    \mathcal{M}_t(U)=\left(\E_t\left[U^{1-\gamma_t}\right]\right)^{\frac{1}{1-\gamma_t}}
\end{equation*}
and the intertemporal aggregators $\{f_t\}_{t\in\mathcal{T}}$ are all weakly increasing, strictly quasi-concave and homogeneous of degree one. These conditions ensure, in conjunction with the constant returns to scale on the production side of the model, that the environments faced by the entrepreneur are homothetic.

\begin{lem}
\label{lem:enthom}
The induced environments $\{(f_t,v_t)\}_{t\in\mathcal{T}}$ are homothetic.
\end{lem}
\begin{proof}
See Appendix \ref{lem:enthomproof}.
\end{proof}

With these ingredients in hand, we study how entrepreneur consumption and savings respond to the following five comparative statics:

\begin{enumerate}
\setcounter{enumi}{10}
    \item\label{item:cs11} The production technology becomes less productive: the distribution of $z_t$, denoted by $F_{z_t}$, becomes $\tilde F_{z_t}$ with $F_{z_t}\succeq_\mathrm{FOSD} \tilde F_{z_t}$.
    \item\label{item:cs12} Wage rates increase: the distribution of $\nu_t$, denoted by $F_{\nu_t}$, becomes $\tilde F_{\nu_t}$ with $F_{\nu t}\succeq_\mathrm{FOSD} \tilde{F}_{\nu_t}$.
    \item\label{item:cs13} Depreciation rates increase: the distribution of $\delta_t$, denoted by $F_{\delta_t}$, becomes $\tilde F_{\delta_t}$ with $F_{\delta_t}\succeq_\mathrm{FOSD} \tilde{F}_{\delta_t}$.
    \item\label{item:cs14} Depreciation rates become riskier: $F_{\delta_t}$ becomes $\tilde{F}_{\delta_t}$ with $F_{\delta_t}\succeq_\mathrm{SOSD} \tilde{F}_{\delta_t}$.
    \item\label{item:cs15} The capital tax rate increases: $\tau_t$ becomes $\tilde \tau_{t}$ with $\tilde\tau_t\ge\tau_t$.
\end{enumerate}

\begin{prop}
\label{prop:ent}
Comparative statics \ref{item:cs11}--\ref{item:cs15} lower the agent's continuation value function at all previous dates.
\end{prop}
\begin{proof}
See Appendix \ref{prop:entproof}.
\end{proof}

Thus, as the environment is homothetic, it follows by Corollary \ref{cor:hom} that consumption decreases in response to any of these shocks if and only if $\psi\le 1$. This result underscores the critical role of the EIS (and its relationship with unity) in structural asset pricing models with entrepreneurs \citep[see \textit{e.g.,}][]{ditella2017}.

\section{Implications for Identification and Estimation of EIS}
\label{sec:implications}
As we discussed in the introduction, the relationship between the EIS and unity is critical for understanding various qualitative and quantitative properties of dynamic models in both macroeconomics and finance. Moreover, there is no empirical consensus on the value of the EIS. For example, \cite{HHIR2015} collect 2,735 estimates of EIS from 169 published studies and find that the mean and standard deviation of published estimates of EIS from 33 articles in the top 5 economics journals are 0.5 and 1.4, respectively.

Our main theoretical results can be expressed in more statistical language as sign-identifying $\text{REMV}\times \text{EIS}-1$ from the sign of consumption responses to exogenous shocks. Under homotheticity, the sign of consumption responses sign-identifies $\text{EIS}-1$. We can operationalize these theoretical results to provide a roadmap for applied researchers to empirically estimate the sign of $\text{EIS}-1$, and to point-identify the EIS under stronger assumptions.

First, without additional structural assumptions, our results clarify that the complications posed by the REMV generally prevent identification of the EIS from consumption responses alone. This observation by itself provides a window for understanding why various empirical strategies may fail to recover the EIS as either the shocks considered or populations of interest may have non-unit REMV. Thus, for this roadmap, we suppose that the researcher is willing to assume that the agents' problems are approximately homothetic. 

Second, suppose that we have access to a dataset of individuals indexed by $i$, potentially with a panel dimension indexed by $t$, that includes data on consumption $c_{it}$, total financial wealth $w_{it}$, and some aggregate or idiosyncratic shifters of investment opportunities or preferences $\alpha_t$ and $\alpha_{it}$.

Third, we can use our theoretical results to derive the following formula for the EIS:
\begin{cor}
\label{cor:point}
Suppose that the environment is homothetic and strongly regular.\footnote{Homotheticity is in the sense of Definition \ref{def:hom} and strong regularity is in the sense of Definition \ref{def:reg}.} If we have a shifter $x\in\{\alpha_t,\alpha_{it}\}$ of the marginal value of wealth $g_{it}(\alpha_{t},\alpha_{it})$, then
\begin{equation*}
    \psi_{it} = 1 -\frac{\frac{\partial\log\frac{c_{it}}{w_{it}-c_{it}}}{\partial x}}{\frac{\partial\log g_{it}}{\partial x}}.
\end{equation*}
\end{cor}
\begin{proof}
See Appendix \ref{cor:pointproof}.
\end{proof}
This suggests an instrumental-variables-like empirical strategy that is valid for an arbitrary shifter of investment opportunities with the numerator $\frac{\partial\log\frac{c_{it}}{w_{it}-c_{it}}}{\partial x}$ representing the \textit{reduced form} and the denominator $\frac{\partial\log g_{it}}{\partial x}$ representing the \textit{first stage}.

Fourth, this strategy can be operationalized. If we are willing to assume that $\psi_{it}$ depends on some set of observable characteristics, one could estimate the reduced-form elasticity $\frac{\partial \log \frac{c_{it}}{w_{it}-c_{it}}}{\partial x}$ directly from the data within groups with the same (or sufficiently similar) observables.  If we can also estimate the magnitude of the first stage $\frac{\partial \log g_{it}}{\partial  x}$, then EIS is point-identified. This could be achieved by finding shocks and settings in which $\frac{\partial \log g_{it}}{\partial  x}$ is known or estimable. For example, within the setting of our portfolio allocation application, in Appendix \ref{lem:homoproof} we provide a formula for $g$ in terms of the agent's aggregator, their risk aversion, and the investment opportunities available to them. Under structural assumptions on these objects, $\frac{\partial \log g_{it}}{\partial  x}$ is obtainable, and point-identification can be achieved.

However, even when this is not feasible, the researcher can still sign-identify $\psi_{it}-1$ if they are willing to assume the sign of $\frac{\partial \log g_{it}}{\partial  x}$ based on knowledge of the shock under consideration. Concretely, suppose without loss of generality that $\frac{\partial \log g_{it}}{\partial  x}>0$. Then we have that:
\begin{equation*}
    \sgn\left(\frac{\partial \log \frac{c_{it}}{w_{it}-c_{it}}}{\partial x}\right) =     \sgn(1-\psi_{it}) 
\end{equation*}
and the sign of the consumption response identifies the relationship of the EIS with unity.

This strategy is, however, subject to the following two caveats. First, it relies on the assumption that the REMV is known to equal one. As a result, it is likely to be most applicable to populations for which permanent income from human capital is relatively unimportant relative to financial wealth (in line with Lemma \ref{lem:remv}). Otherwise precautionary savings motives may induce non-unit REMV and prevent identification. For example, if changes in investment opportunities take place alongside changes in labor income risk and permanent income is non-negligible relative to financial wealth, then the REMV will not equal one. Consequently, this strategy is likely to apply to older and wealthier groups of households -- which make up the bulk of participants in financial markets -- but may struggle to identify the EIS for younger and poorer households. 

Nevertheless, even when labor income risk makes the REMV non-unitary, we can extend the above strategy to identify the value of value of $\text{EIS}-1$ under certain conditions. Lemma \ref{lem:remv} implies that the REMV is greater than one if and only if the elasticity of the marginal value of wealth exceeds the elasticity of the marginal value of permanent income.\footnote{Mathematically speaking, Lemma \ref{lem:remv} implies that: $\varepsilon\ge1 \iff 1\ge\frac{v_{p\alpha}/v_p}{v_{w\alpha}/v_w}$.}
Thus, if we observe a decrease in consumption in response to an adverse shock to investment opportunities and we are willing to suppose the previous condition holds, then we know that $\psi\leq1/\varepsilon\leq 1$. Hence, a researcher can leverage their knowledge of the shock under consideration (perhaps through the lens of a structural model) to provide bounds on the REMV that allow set identification of the EIS from the sign of consumption responses alone.

Second, households must actually perceive and act upon the shocks, so that they are relevant and the first stage is non-zero. If households are inattentive (or otherwise cognitively constrained) or face large adjustment costs, it is possible that the shocks identified by the researcher may not influence household behavior, preventing identification. There are at least two possibilities to overcome this limitation: the researcher could verify by a direct survey that a household is aware of the shock under consideration; or they could consider large shocks which are likely to have large costs to ignore.

\section{Conclusion}
\label{sec:conclusion}
In this paper, we study consumption-savings problems with general recursive preferences. We characterize the sign of the consumption responses to arbitrary shocks in terms of whether the product of two sufficient statistics, the EIS and the REMV, is greater or less than one. In homothetic environments, the REMV is always one, and the sign of consumption responses is characterized solely by the relationship of the EIS with unity. This allows us to derive a range of comparative statics in applications to portfolio allocation, consumption-savings problems with income risk, and entrepreneurial investment.

In more empirical language, our results sign-identify $\text{EIS}-1$ with the sign of the consumption response to a variety of shocks under homotheticity. This is important for two reasons. First, this relationship is critical for the qualitative and quantitative predictions of dynamic models in macroeconomics and finance as well as their normative implications. Second, there is a large amount of uncertainty regarding this relationship empirically. Finally, under additional structural assumptions, our formulae for the consumption responses to shocks can be used to identify the EIS even when homotheticity fails.

\begin{appendices}
\section{Omitted Proofs}
\subsection{Proof of Proposition \ref{prop:main_two}}
\label{prop:main_twoproof}
\begin{proof}
The agent's problem can be stated as:
\begin{equation*}
    \max_{c,v\in\R_{+}} f(c,\rho v) \quad \text{subject to}\quad v=R_f(e_1-c)+e_2.
\end{equation*}
As we have assumed an Inada condition on $f$, we can ignore non-negativity constraints. The following first-order condition is therefore necessary for optimality:
\begin{equation*}
    f_c(c, \rho R_f(e_1-c) + \rho e_2) -\rho R_f f_v(c, \rho R_f(e_1-c) + \rho e_2) = 0.
\end{equation*}
Applying the implicit function theorem, we can compute:
\begin{align*}
    \frac{\dd \log \left(\frac{c}{v}\right)}{\dd \rho} &=\frac{\frac{\partial c}{\partial \rho}}{c} - \frac{R_f(e_1-c)+e_2-R_f\frac{\partial c}{\partial \rho}}{R_f(e_1-c)+e_2}, & \frac{\dd \log \left(\frac{f_c}{f_v}\right)}{\dd \rho} &= 1,
\end{align*}
\begin{align*}
    \frac{\dd \log \left(\frac{c}{v}\right)}{\dd R_f} &= \frac{\frac{\partial c}{\partial R_f}}{c}-\frac{-R_f\frac{\partial c}{\partial R_f}+ (e_1-c)}{ R_f(e_1-c) + e_2}, & \frac{\dd \log \left(\frac{f_c}{f_v}\right)}{\dd R_f} &= \frac{1}{R_f}.
\end{align*}
Thus, by the definition of the EIS, rearranging yields:
\begin{align*}
    \left(\frac{1}{c}+\frac{R_f}{R_f(e_1-c)+e_2}\right)\frac{\partial c}{\partial \rho} &= 1-\psi,\\
    \left(\frac{1}{c}+\frac{R_f}{R_f(e_1-c)+e_2}\right)\frac{\partial c}{\partial R_f} &= \frac{1}{R_f}\left(\frac{R_f(e_1-c)}{R_f(e_1-c)+e_2}-\psi\right).
\end{align*}
Substituting the definition of $\varepsilon$ completes the proof.
\end{proof}

\subsection{Proof of Theorem \ref{thm1}}
\label{thm1proof}
\begin{proof}
By regularity, $f$ and $v$ are twice continuously differentiable. Moreover, $[0,w]$ is compact. Thus, by the extreme value theorem we have that the maximum is attained and the agent solves:
\begin{equation}
\label{eq:regprob}
    \max_{c\in[0,w]}f(c,v(w-c,\alpha)).
\end{equation}
By strong regularity $c=c(w,\alpha)\in(0,w)$. Thus, any optimal $c$ solves the first-order condition
\begin{equation}
\label{eq:focthm1}
    f_c(c,v(w-c,\alpha))-v_w(w-c,\alpha)f_v(c,v(w-c,\alpha))=0.
\end{equation}
Thus, suppressing all arguments, we can compute:
\begin{equation}
    \frac{\dd}{\dd\alpha}\log\left(\frac{f_c}{f_v}\right)=\frac{\dd}{\dd\alpha}\log v_w = \frac{v_{w\alpha}-v_{ww}c_{\alpha}}{v_w},\label{eq:EIS_denom}
\end{equation}
where all partial derivatives here exist by the hypothesis of strong regularity. In particular, the partial derivative of $c$ with respect to $\alpha$ (which was not assumed to exist) obtains by application of the implicit function theorem with respect to \eqref{eq:focthm1}. We can moreover compute
\begin{equation}
    \frac{\dd}{\dd\alpha}\log\left(\frac{c}{v}\right)=\frac{c_{\alpha}}{c}-\frac{-v_wc_{\alpha}+v_{\alpha}}{v}.\label{eq:EIS_num}
\end{equation}
By \eqref{eq:EIS_denom}, \eqref{eq:EIS_num}, and the definition of EIS in \eqref{eq:EIS_def}, we have
\begin{equation*}
    -\psi = \frac{\frac{c_{\alpha}}{c}-\frac{-v_wc_{\alpha}+v_{\alpha}}{v}}{\frac{v_{w\alpha}-v_{ww}c_{\alpha}}{v_w}},
\end{equation*}
which is equivalent to
\begin{equation*}
    \left(\frac{1}{c}+\frac{v_w}{v}-\psi\frac{v_{ww}}{v_w}\right)c_{\alpha}=\frac{v_{\alpha}}{v}\left(1-\frac{\frac{v_{w\alpha}}{v_w}}{\frac{v_{\alpha}}{v}}\psi\right) = \frac{v_{\alpha}}{v}(1-\varepsilon\psi),
\end{equation*}
where the final equality follows by the definition of the REMV. The final claim that \eqref{eq:thm1} holds  when $v_{ww}\le 0$ follows immediately by noting that $v_w \ge 0$ and $\psi\ge0$.
\end{proof}

\subsection{Proof of Theorem \ref{thm2}}
\label{thm2proof}
\begin{proof}
By regularity, the agent faces problem \eqref{eq:regprob}. Define the function $\tilde{f}(c,\alpha)=f(c,v(w-c,\alpha))$. The constraint set $[0,w]$ is a lattice and does not depend on $\alpha$. Furthermore, $\alpha\in[0,1]$, which is a totally ordered set. As $c\in\R_{+}$, $\tilde{f}$ is quasi-supermodular in $c$. Thus, if $\tilde{f}$ satisfies the strict single-crossing property in $(c,\alpha)$, then by Theorem 4' in \cite{MilgromShannon1994}, any optimal consumption function must be increasing in $\alpha$. By the hypothesis of regularity, $\tilde{f}$ is twice continuously differentiable. Thus, the strict supermodularity condition $\tilde{f}_{c\alpha}>0$ is sufficient for the strict single-crossing property. Taking partial derivatives, this can be expressed as
\begin{equation*}
    \tilde{f}_{c\alpha} =     \left(f_{vc}-f_{vv}v_w\right)v_{\alpha}-f_vv_{\alpha w} > 0,
\end{equation*}
which is equivalent to
\begin{equation}
f_{cv} > \frac{v_{w\alpha}}{v_{\alpha}}f_v+v_{w} f_{vv} = \frac{v_w}{v}\frac{\frac{v_{w\alpha}}{v_w}}{\frac{v_{\alpha}}{v}}f_v+\frac{v_w}{v}vf_{vv}.\label{eq:fcv}
\end{equation}
If well defined, we can rewrite \eqref{eq:fcv} as
\begin{align*}
    1 &> \frac{v_w}{v}\left(\frac{\frac{v_{w\alpha}}{v_w}}{\frac{v_{\alpha}}{v}}f_v+vf_{vv}\right) f_{cv}^{-1} = \frac{v_w}{v}\left(\frac{\frac{v_{w\alpha}}{v_w}}{\frac{v_{\alpha}}{v}}+\frac{vf_{vv}}{f_v}\right) \frac{f_v}{f_{cv}} \\
    &=\frac{\frac{v_w}{v}}{\frac{f_c}{f}}\left(\frac{\frac{v_{w\alpha}}{v_w}}{\frac{v_{\alpha}}{v}}+\frac{vf_{vv}}{f_v}\right) \frac{f_cf_v}{f_{cv}f} = \frac{\frac{v_w}{v}}{\frac{f_c}{f}}\left(\frac{\frac{v_{w\alpha}}{v_w}}{\frac{v_{\alpha}}{v}}+\frac{vf_{vv}}{f_v}\right) \left(\frac{f_{cv}f} {f_cf_v}\right)^{-1},
\end{align*}
completing the proof.
\end{proof}

\subsection{Proof of Corollary \ref{cor:hom}}
\label{cor:homproof}
\begin{proof}
The consumption response is given by \eqref{eq:main_gen}. As $v(w)=g(\alpha) w$, we have that $\frac{v_w}{v}=\frac{1}{w-c}$, $\frac{v_{\alpha}}{v}=\frac{g_{\alpha}}{g}>0$, $v_{ww}=0$, and $\varepsilon=1$. Thus:
\begin{equation}
\label{eq:corform}
    \left(\frac{1}{c}+\frac{1}{w-c}\right)c_{\alpha}=\frac{g_{\alpha}}{g}(1-\psi)
\end{equation}
and \eqref{eq:cor1} follows immediately by noting that $g_{\alpha},g>0$.
\end{proof}

\subsection{Proof of Corollary \ref{cor:hommcs}}
\label{cor:hommcsproof}
\begin{proof}
By Theorem \ref{thm2}, we have that any possible consumption function is increasing in $\alpha$ if \eqref{eq:mcscond} holds. 
Under homotheticity, by \eqref{eq:epshom}, we have that $\frac{\frac{v_{w\alpha}}{v_w}}{\frac{v_{\alpha}}{v}}=1$. Moreover, $\frac{v_w}{v}=\frac{1}{w-c}$. Further, by homogeneity (of degree one) of the aggregator and Euler's theorem we have that $f=cf_c+vf_v$. This implies that $f_{vv}=-\frac{c}{v}f_{cv}$. Substituting these observations yields
\begin{align*}
\frac{\frac{v_w}{v}}{\frac{f_c}{f}}\left(\frac{\frac{v_{w\alpha}}{v_w}}{\frac{v_{\alpha}}{v}}+\frac{vf_{vv}}{f_v}\right)\left(\frac{f_{cv}f}{f_cf_v}\right)^{-1} &= \frac{1}{w-c}\left(\frac{f_c}{f}\right)^{-1}\left(1+\frac{v\left(-\frac{c}{v}f_{cv}\right)}{f_v}\right)\left(\frac{f_{cv}f}{f_cf_v}\right)^{-1} \\
&=\frac{1}{w-c}\left(1-c\frac{f_{cv}}{f_v}\right)\left(\frac{f_{cv}}{f_v}\right)^{-1}=\frac{1}{w-c}\left(\frac{f_v}{f_{cv}}-c\right).
\end{align*}
Hence, the sufficient condition \eqref{eq:mcscond} becomes
\begin{equation*}
    \frac{1}{w-c}\left(\frac{f_v}{f_{cv}}-c\right) < 1.
\end{equation*}
Rewriting this yields
\begin{equation*}
   1> \frac{1}{w}\frac{f_v}{f_{cv}}=\frac{\frac{1}{w}}{\frac{f_c}{f}}\frac{f_cf_v}{f_{cv}f}= \frac{\frac{1}{w}}{\frac{f_c}{f}}\left(\frac{f_{cv}f}{f_cf_v}\right)^{-1},
\end{equation*}
Completing the proof.
\end{proof}

\subsection{Proof of Proposition \ref{prop:risk}}
\label{prop:riskproof}
\begin{proof}
By definition, $v_t(w)=\max_{\theta\in\Theta_t}\phi_t^{-1}\left(\E_t\left[\phi_t\left(V_{t+1}(w')\right)\right]\right)$, with $w'=R_{t+1}(\theta) w$. We prove the five comparative statics in turn.
\begin{enumerate}
    \item Risk aversion increases: $\phi_t$ changes to $\tilde{\phi}_t$, where $g=\tilde{\phi}_t\circ \phi_t^{-1}$ is increasing and concave. See that we can write:
    \begin{align*}
        \tilde\phi_t(\tilde v_t(w))&=\max_{\theta_t\in\Theta_t}\E_t\left[\tilde\phi_t\left(V_{t+1}(w')\right)\right] =\max_{\theta_t\in\Theta_t}\E_t\left[g\circ\phi_t\left(V_{t+1}(w')\right)\right]\\
        &\le \max_{\theta_t\in\Theta_t} g\left(\E_t\left[\phi_t\left(V_{t+1}(w')\right)\right] \right) = g\left(\max_{\theta_t\in\Theta_t}\E_t\left[\phi_t\left(V_{t+1}(w')\right)\right]\right)\\
        &= g\circ \phi_t(v_t(w)) = \tilde \phi_t(v_t(w)),
    \end{align*}
    where the inequality follows by Jensen's inequality. This implies that $\tilde v_t(w)\le v_t(w)$.
    
    \item The investment opportunity set shrinks: $\Theta_t$ changes to $\tilde{\Theta}_t\subset \Theta_t$. See that:
    \begin{equation*}
        \tilde v_t(w) = \max_{\theta\in\tilde\Theta_t}\phi_t^{-1}\left(\E_t\left[\phi_t\left(V_{t+1}(w')\right)\right]\right) \le \max_{\theta\in\Theta_t}\phi_t^{-1}\left(\E_t\left[\phi_t\left(V_{t+1}(w')\right)\right]\right) =  v_t(w)
    \end{equation*}

    \item The portfolio returns become lower:
    $\{R_{t+1}(\theta)\}_{\theta\in\Theta_{t}}$ changes to $\{\tilde R_{t+1}(\theta)\}_{\theta\in\Theta_{t}}$ where $R_{t+1}(\theta)\succeq_\mathrm{FOSD}\tilde R_{t+1}(\theta)$ for all $\theta\in\Theta_t$. We can write:
    \begin{align}
        \phi_t(\tilde v_t(w)) &= \max_{\theta\in\Theta_t}\E_t\left[\phi_t\left(V_{t+1}(\tilde R_{t+1}(\theta)w)\right)\right] = \E_t\left[\phi_t\left(V_{t+1}(\tilde R_{t+1}(\tilde\theta^*)w)\right)\right] \notag \\
        &\le \E_t\left[\phi_t\left(V_{t+1}( R_{t+1}(\tilde\theta^*)w)\right)\right] \le \max_{\theta\in\Theta_t}\E_t\left[\phi_t\left(V_{t+1}( R_{t+1}(\theta)w)\right)\right] \notag \\
        & = \phi_t( v_t(w)), \label{eq:lower}
    \end{align}
    where the first inequality follows as $R_{t+1}(\theta) \succeq_\mathrm{FOSD} \tilde R_{t+1}(\theta)$ for all $\theta\in\Theta_t$, $\phi_t$ is increasing, $V_{t+1}$ is increasing in wealth, and the second by the definition of the maximum.

    \item The portfolio returns become riskier: $\{R_{t+1}(\theta)\}_{\theta\in\Theta_{t}}$ changes to $\{\tilde R_{t+1}(\theta)\}_{\theta\in\Theta_{t}}$ where $R_{t+1}(\theta)\succeq_\mathrm{SOSD}\tilde R_{t+1}(\theta)$ for all $\theta\in\Theta_t$. This follows by exactly the same chain of inequalities as \eqref{eq:lower}, but where the first inequality follows as as $R_{t+1}(\theta) \succeq_\mathrm{SOSD} \tilde R_{t+1}(\theta)$ for all $\theta\in\Theta_t$, $\phi_t$ is concave and $V_{t+1}$ is concave by hypothesis.

    \item Future consumption expenditure becomes less valuable: the aggregator shifts from $f_{t+1}$ to $\tilde{f}_{t+1}$ where $\tilde{f}_{t+1}(c,v) = f_{t+1}(c/g(c,v),v)$, where $g(c,v) \ge 1$ for all $(c,v)$. We can write the period $t+1$ value function as
    \begin{align*}
        \tilde V_{t+1}(w)&=\max_{c\in[0,w]}\tilde{f}_{t+1}(c, v_{t+1}(w-c))\\
        &= \max_{c\in[0,w]}f_{t+1}\left(\frac{c}{g(c,v_{t+1}(w-c)}, v_{t+1}(w-c)\right) \\
        &= f_{t+1}\left(\frac{c^*}{g(c^*,v_{t+1}(w-c^*))}, v_{t+1}(w-c^*)\right) \\
        &\le f_{t+1}\left(c^*, v_{t+1}(w-c^*)\right) \\
        &\le \max_{c\in[0,w]} f_{t+1}(c, v_{t+1}(w-c)) = V_{t+1}(w).
    \end{align*}
    Thus, we have that
\begin{equation*}
    \tilde v_t(w)=\max_{\theta\in\Theta_t}\phi_t^{-1}\left(\E_t\left[\phi_t\left(\tilde V_{t+1}(w')\right)\right]\right) \le \max_{\theta\in\Theta_t}\phi_t^{-1}\left(\E_t\left[\phi_t\left( V_{t+1}(w')\right)\right]\right) = v_t(w).\qedhere
\end{equation*}
\end{enumerate}
\end{proof}

\subsection{Proof of Lemma \ref{lem:homo}, Explicit Epstein-Zin solution, and Extension to Death with Bequests}
\label{lem:homoproof}
We first prove Lemma \ref{lem:homo}.

\begin{proof}
To show that the induced environments are homothetic, by Definition \ref{def:hom} we need to show that: $f_t$ is homogeneous of degree one and that $v(w,\alpha)=g(\alpha)w$ with $g>0$. We have the first of these by assumption. Thus, it suffices to show the second. We do this by first establishing that $V_{t}(w)=b_{t}w$ for some random variables $b_{t}$ for all $t\in\mathcal{T}$. To this end, observe in period $T$ that $u_T(c_T)=b_Tc_T$ for $b_T>0$. Thus, we have that $V_T(w)=b_Tw$. Proceeding inductively, suppose that $V_{t+1}(w)=b_{t+1}w$ for some $b_{t+1}>0$. We have that
\begin{align*}
    V_t(w)&=\max_{c\in[0,w],\theta\in\Theta_t}f_t\left(c,\E_t\left[V_{t+1}\left(R_{t+1}(\theta)(w-c)\right)^{1-\gamma_t}\right]^{\frac{1}{1-\gamma_t}}\right) \\
    &=\max_{c\in[0,w],\theta\in\Theta_t}f_t\left(c,\E_t\left[\left(b_{t+1}R_{t+1}(\theta)(w-c)\right)^{1-\gamma_t}\right]^{\frac{1}{1-\gamma_t}}\right) \\
    &= \max_{c\in[0,w]}f_t\left(c,\max_{\theta\in\Theta_t}\E_t\left[\left(b_{t+1}R_{t+1}(\theta)(w-c)\right)^{1-\gamma_t}\right]^{\frac{1}{1-\gamma_t}}\right) \\
    &= \max_{c\in[0,w]}f_t\left(c,(w-c)\max_{\theta\in\Theta_t}\E_t\left[\left(b_{t+1}R_{t+1}(\theta)\right)^{1-\gamma_t}\right]^{\frac{1}{1-\gamma_t}}\right) \\
    &= \max_{\tilde c\in[0,1]}wf_t\left(\tilde{c},(1-\tilde c)\max_{\theta\in\Theta_t}\E_t\left[\left(b_{t+1}R_{t+1}(\theta)\right)^{1-\gamma_t}\right]^{\frac{1}{1-\gamma_t}}\right) \\
    &= w\max_{\tilde c\in[0,1]}f_t\left(\tilde{c},(1-\tilde c)\max_{\theta\in\Theta_t}\E_t\left[\left(b_{t+1}R_{t+1}(\theta)\right)^{1-\gamma_t}\right]^{\frac{1}{1-\gamma_t}}\right) \\
    &=b_tw,
\end{align*}
where
\begin{equation}
\label{eq:beq}
    b_t = \max_{\tilde c\in[0,1]}f_t\left(\tilde{c},(1-\tilde c)\max_{\theta\in\Theta_t}\E_t\left[\left(b_{t+1}R_{t+1}(\theta)\right)^{1-\gamma_t}\right]^{\frac{1}{1-\gamma_t}}\right) >0.
\end{equation}

The first line is by definition. The second is by the induction hypothesis. The third follows as $f_t$ is upper semi-continuous, the set of portfolios is compact, and returns are bounded. The fourth follows by identity. The fifth follows by homogeneity of degree one of $f_t$. The sixth is by identity. The seventh is by definition. To complete the proof, we now observe that this implies linearity of the continuation value functions in wealth:
\begin{align*}
    v_t(w)&=\max_{\theta\in\Theta_t}\E_t\left[V_{t+1}\left(R_{t+1}(\theta)w\right)^{1-\gamma_t}\right]^{\frac{1}{1-\gamma_t}} \\
    &=\max_{\theta\in\Theta_t}\E_t\left[\left(b_{t+1}R_{t+1}(\theta)w\right)^{1-\gamma_t}\right]^{\frac{1}{1-\gamma_t}} \\
    &=w\max_{\theta\in\Theta_t}\E_t\left[\left(b_{t+1}R_{t+1}(\theta)\right)^{1-\gamma_t}\right]^{\frac{1}{1-\gamma_t}}
\end{align*}
Thus, if we set $g(\alpha)=\max_{\theta\in\Theta_t}\E_t\left[\left(b_{t+1}R_{t+1}(\theta)\right)^{1-\gamma_t}\right]^{\frac{1}{1-\gamma_t}}>0$, we are done.
\end{proof}

We now use this Lemma to derive an explicit solution for the consumption function when the agent has Epstein-Zin preferences.

\begin{prop}
\label{prop:ezsol}
With Epstein-Zin preferences, the optimal consumption function is
\begin{equation*}
    c_t(w)=(1-\beta)^{\psi}b_t^{1-\psi} w,
\end{equation*}
where $b_t$ is defined recursively by
\begin{equation*}
    b_t = \begin{cases}
    \left((1-\beta)^\psi+\beta^\psi\left(\max_{\theta\in\Theta_t}\E_t\left[\left(b_{t+1}R_{t+1}(\theta)\right)^{1-\gamma}\right]^{\frac{1}{1-\gamma}}\right)^{\psi-1}\right)^\frac{1}{\psi-1}, & \psi\neq 1,\\
    (1-\beta)^{1-\beta}\beta^\beta\left(\max_{\theta\in\Theta_t}\E_t\left[\left(b_{t+1}R_{t+1}(\theta)\right)^{1-\gamma}\right]^{\frac{1}{1-\gamma}}\right)^\beta, & \psi=1.
    \end{cases}
\end{equation*}
\end{prop}
\begin{proof}
Nearly identical to \citet[Corollary 7]{Toda2014JET}.
\end{proof}

We now provide conditions under which there is a unique optimal portfolio (Remark \ref{rem:optport}) and extend Lemma \ref{lem:homo} to allow for death and bequest motives (Remark \ref{rem:bequest}).

\begin{remtri}
\label{rem:optport}
Since by assumption the aggregator $f$ is strictly quasi-concave, the optimal consumption rule $\tilde{c}_t$ is unique.  If the portfolio set is $\Theta_t$ finite or convex and there are no redundant assets, then the optimal portfolio is unique.
\end{remtri}

\begin{remtri}
\label{rem:bequest}
Here, we extend the model to incorporate stochastic death and bequests. In each period $t<T$, if still alive, the agent dies with probability $\delta_t$. In period $t=T$, if the agent is still alive, the agent dies for sure. If the agent dies with wealth $w$ in period $t$, then their terminal utility function over bequests is of the form $u_{dt}(w)=b_{dt}(w)$ for some random variable $b_{dt}>0$. In this case, the Bellman equation becomes
\begin{align*}
V_t(w)=&\max_{c\in[0,w],\theta\in\Theta_t}f_t\Bigg(c, \Big[\delta_t \E_t\left[ (b_{d,t+1}R_{t+1}(\theta)(w-c))^{1-\gamma_t}\right]  \notag \\
&+ (1-\delta_t) \E_t\left[V_{t+1}(R_{t+1}(\theta)(w-c))^{1-\gamma_t}\right]\Big]^\frac{1}{1-\gamma_t}\Bigg).
\end{align*}
The extension of Lemma \ref{lem:homo} to this case is immediate. It is also a straightforward extension of Proposition \ref{prop:risk} that adverse changes in the distribution of $b_{d,t+1}$ in the sense of second-order stochastic dominance reduce the continuation value function at all previous dates. 
\end{remtri}

\subsection{Proof of Proposition \ref{prop:ambiguity}}
\label{prop:ambiguityproof}
\begin{proof}
We prove the two comparative statics in turn.
\begin{enumerate}
\setcounter{enumi}{5}
    \item The agent becomes more ambiguity averse in the smooth environment: $\varphi_t$ changes to $\tilde{\varphi}_t$, where $g=\tilde{\varphi}_t\circ \varphi_t^{-1}$ is increasing and concave. We can express the continuation value function as
    \begin{equation*}
        \tilde\varphi_t(\tilde v_t(w))=\max_{\theta\in\Theta_t}\E_{\mu_t}[\tilde\varphi_t(\phi_t^{-1}(\E_{\pi_t}[\phi_t(V_{t+1}(w'))]))],
    \end{equation*}
    where $w'=R_{t+1}(\theta)w$. Hence, application of the same steps as in the proof of comparative static \ref{item:cs1} in Proposition \ref{prop:risk} yields the result.
    
    \item The agent considers more prior distributions in the infinitely ambiguity averse environment: $\mathcal{P}_t$ changes to $\tilde{\mathcal{P}}_t$ with $\tilde{\mathcal{P}}_t\supset \mathcal{P}_t$. As the minimization under $\tilde{\mathcal{P}}_t$ is taken over a larger set, we have that
    \begin{align*}
\tilde v_t(w)&=\max_{\theta\in\Theta_t}\phi_t^{-1}\left(\min_{\pi_t\in \tilde{\mathcal{P}}_t}\E_{\pi_t}[\phi_t(V_{t+1}(w'))]\right) \\
&\le \max_{\theta\in\Theta_t}\phi_t^{-1}\left(\min_{\pi_t\in \mathcal{P}_t}\E_{\pi_t}[\phi_t(V_{t+1}(w'))]\right) = v_t(w). \qedhere
  \end{align*}
\end{enumerate}
\end{proof}

\subsection{Proof of Proposition \ref{prop:cs}}
\label{prop:csproof}
\begin{proof}
Observe that the continuation value function in this instance is given by
\begin{equation*}
    v_t(w;p)=\max_{\theta\in\Theta_t}\phi_t^{-1}\left(\E_t\left[\phi_t\left(V_{t+1}\left(R_{t+1}(\theta)w+p\tau_{t+1}\eta_{t+1}\right)\right)\right]\right).
\end{equation*}
We prove the three comparative statics in turn.
\begin{enumerate}
\setcounter{enumi}{7}
    \item The agent's income permanent or transitory income falls in the sense of FOSD. Observe that under either of these changes that the distribution of the random variable $R_{t+1}(\theta)w+p\tau_{t+1}\eta_{t+1}$ falls in the sense of FOSD. Thus, the same steps as in the proof of comparative static \ref{item:cs3} in Proposition \ref{prop:risk} establish the result.
    
    \item The household's investment or hedging opportunities shrink: $\Theta_t$ changes to $\tilde{\Theta}_t\subset \Theta_t$. The same steps as in the proof of comparative static \ref{item:cs2} in Proposition \ref{prop:risk} establish the result.
    
    \item The agent's income becomes riskier in the sense of SOSD. Observe that under either of these changes that the distribution of the random variable $R_{t+1}(\theta)w+p\tau_{t+1}\eta_{t+1}$ falls in the sense of SOSD. Thus, the same steps as in the proof of comparative static \ref{item:cs4} in Proposition \ref{prop:risk} establish the result. \qedhere
\end{enumerate}
\end{proof}

\subsection{Proof of Lemma \ref{lem:remv}}
\label{lem:remvproof}
\begin{proof}
To show this result, we first establish that the continuation value functions $\{v_t\}_{t\in\mathcal{T}}$ are homogeneous of degree one. We do this by backward induction. Consider the terminal period $T$. We have that $V_T(w,p)=u_T(w)=b_Tw$, which is homogeneous of degree one. By the definition of the continuation value function in period $T-1$, we have
\begin{align*}
        v_{T-1}(\lambda w,\lambda p)& =\max_{\theta\in\Theta_{T-1}}\E_{T-1}\left[V_T\left(R_{T}(\theta)\lambda w+\lambda p\tau_{T}\eta_{T},\lambda p\tau_{T}\eta_{T}\right)^{1-\gamma_{T-1}}\right]^{\frac{1}{1-\gamma_{T-1}}} \\
        &=\max_{\theta\in\Theta_{T-1}}\E_{T-1}\left[\lambda^{1-\gamma_{T-1}} V_T\left(R_{T}(\theta) w+ p\tau_{T}\eta_{T}, p\tau_{T}\eta_{T}\right)^{1-\gamma_{T-1}}\right]^{\frac{1}{1-\gamma_{T-1}}} \\
        &=\lambda v_{T-1}(w,p).
\end{align*}
Thus, $v_{T-1}$ is homogeneous of degree one.

Proceeding inductively, suppose that $v_{t}$ is homogeneous of degree one. We wish to show that $v_{t-1}$ is homogeneous of degree one. We first show that $V_t$ is homogeneous of degree one:
\begin{align*}
    V_t(\lambda w,\lambda p) &=\max_{c\in[0,\lambda w]} f_t\left(c,v_t(\lambda w-c,\lambda p)\right) \\
    &= \max_{c\in[0,\lambda w]} \kappa f_t\left(\frac{c}{\kappa},\frac{v_t(\lambda w-c,\lambda p)}{\kappa}\right) \\
    &= \kappa\max_{c\in[0,\lambda w]}  f_t\left(\frac{c}{\kappa},v_t\left(\frac{\lambda w-c}{\kappa},\frac{\lambda p}{\kappa}\right)\right) \\
    &= \lambda \max_{c\in[0,\lambda w]}  f_t\left(\frac{c}{\lambda},v_t\left(\frac{\lambda w-c}{\lambda},\frac{\lambda p}{\lambda}\right)\right) \\
    &=\lambda \max_{\tilde c\in[0, w]}  f_t\left(\tilde c,v_t\left(w-\tilde c,p\right)\right) \\
    &= \lambda V_t(w,p),
\end{align*}
where the first equality is by definition, the second is by homogeneity of degree one of $f_t$, the third is by homogeneity of degree one of $v_t$, the fourth is by setting $\kappa=\lambda$, the fifth is by defining $\tilde c=\frac{c}{\lambda}$, and the last is by the definition of the value function.

We now use this fact to show that $v_{t-1}$ is homogeneous of degree one. Note that
\begin{align*}
    v_{t-1}(\lambda w,\lambda p)&=\max_{\theta\in\Theta_{t-1}}\E_{t-1}\left[V_t\left(R_{t}(\theta)\lambda w+\lambda p\tau_{t}\eta_{t},\lambda p\tau_{t}\eta_{t}\right)^{1-\gamma_{t-1}}\right]^{\frac{1}{1-\gamma_{t-1}}} \\
    &= \max_{\theta\in\Theta_{t-1}}\E_{t-1}\left[\lambda^{1-\gamma_{t-1}}V_t\left(R_{t}(\theta) w+p\tau_{t}\eta_{t}, p\tau_{t}\eta_{t}\right)^{1-\gamma_{t-1}}\right]^{\frac{1}{1-\gamma_{t-1}}} \\
    &= \lambda v_{t-1}(w,p).
\end{align*}
We now use this to derive the claimed formula for the REMV. By Euler's theorem, we have that $v=v_w w+ v_p p$. Thus, we have the required expression
\begin{equation*}
\varepsilon=\frac{\frac{\partial \log v_w}{\partial\alpha}}{\frac{\partial \log v}{\partial\alpha}} = \frac{v_{w\alpha}}{v_w}\frac{v_w w+ v_p p}{v_{w\alpha} w + v_{p\alpha} p} =  \frac{1+\frac{p}{w}\frac{v_p}{v_w}}{1+\frac{p}{w}\frac{v_{p\alpha}}{v_{w\alpha}}}. \qedhere
\end{equation*}
\end{proof}

\subsection{Proof of Lemma \ref{lem:enthom}}
\label{lem:enthomproof}
\begin{proof}
We can re-express the model as the same as the one we developed in Section \ref{sec:risk}. Once this is done, the result follows by verifying the hypotheses of Lemma \ref{lem:homo}. Define $\tilde\theta_t=(\theta_t,a_t,b_t,k_t)\in\Tilde\Theta_t$ and observe that we can write
\begin{equation*}
    w_{t+1}=R_{t+1}(\tilde\theta_t)(w_t-c_t),
\end{equation*}
where we have simplified away $y_{t+1}$ and $l_{t+1}$ by observing that any optimal $l_{t+1}$ is pinned down immediately by $k_{t+1}$, $g_{t+1}$, $\nu_{t+1}$, $z_{t+1}$ by solving the cost minimization problem. Thus, the model reduces to our portfolio allocation problem. Moreover, the aggregator is CRRA, terminal utility is linear, $f_t$ is increasing, strictly quasi-concave and homogeneous of degree one. The returns are moreover bounded and $\Tilde\Theta_t$ is compact. Thus, Lemma \ref{lem:homo} yields the result.
\end{proof}

\subsection{Proof of Proposition \ref{prop:ent}}
\label{prop:entproof}
\begin{proof}
Observe by the proof of Lemma \ref{lem:enthomproof} that the continuation value function can be re-expressed in the form of the homothetic case of the model from Section \ref{sec:risk}. Thus, by Lemma \ref{lem:homo}, we have that
\begin{equation*}
    v_t(w)=b_t w,
\end{equation*}
where $b_t$ is given by \eqref{eq:beq}. We now prove the five comparative statics in turn.
\begin{enumerate}
\setcounter{enumi}{10}
    \item The production technology becomes less productive in the sense of FOSD. This follows by observing that this reduces returns and the same steps as comparative static \ref{item:cs3} in the proof of Proposition \ref{prop:risk}.
    \item Wage rates increase in the sense of FOSD. This follows by the same argument as \ref{item:cs11}.
    \item Depreciation rates increase in the sense of FOSD. This follows by the same argument as \ref{item:cs11}.
    \item Depreciation rates become riskier in the sense of SOSD. This increases the riskiness of returns and follows by the same argument as comparative static \ref{item:cs4} in Proposition \ref{prop:risk}.
    \item The capital tax rate increases. This follows by the same argument as \ref{item:cs11}. \qedhere
\end{enumerate}
\end{proof}

\subsection{Proof of Corollary \ref{cor:point}}
\label{cor:pointproof}
\begin{proof}
By \eqref{eq:corform} in the proof of Corollary \ref{cor:hom}, we have
\begin{equation*}
    \left(\frac{1}{c}+\frac{1}{w-c}\right)c_{\alpha} = \frac{g_{\alpha}}{g}(1-\psi).
\end{equation*}
Observing that the LHS is simply $\frac{\partial\log\frac{c}{w-c}}{\partial\alpha}$ and $\frac{g_{\alpha}}{g}=\frac{\partial \log g}{\partial\alpha}$, the result follows immediately.
\end{proof}

\section{Sufficient Conditions for Strong Regularity}
\label{ap:suffcond}

The conditions for (strong) regularity require primitive conditions on $f$ as well as non-primitive conditions on $v$ and $c$. In this Appendix, we provide verifiable conditions on the \textit{setting} $\{u_T,(\Omega,\mathcal{F},P),\{f_t,\mathcal{M}_t,\Theta_t,W_t\}_{t\in\mathcal{T}}\}$ such that the induced environments $\{(f_t,v_t)\}_{t\in\mathcal{T}}$ are strongly regular. These take the form of technical restrictions on the probability space with respect to which random variables are defined and standard interiority and smoothness conditions on the evolution of wealth and both intratemporal and intertemporal aggregation.

\begin{defn}
A setting $\{u_T,(\Omega,\mathcal{F},P),\{f_t,\mathcal{M}_t,\Theta_t,W_t\}_{t\in\mathcal{T}}\}$ is discrete, interior, and smooth (DIS) if
\begin{enumerate}
    \item the state space $\Omega$ is discrete,
    
    \item all intertemporal aggregators $\{f_t\}_{t\in\mathcal{T}}$ and the terminal utility function $u_T$ are strictly increasing, infinitely continuously differentiable, normalized in the sense that $f_t(0,0)=u_T(0)=0$, and satisfy the Inada conditions that $\lim_{c\to0}f_{tc}(c,v)=\infty$, $\lim_{v\to0}f_{tv}(c,v)=\infty$, and $\lim_{c\to0}u_{Tc}(c)=\infty$,
    
    \item the certainty equivalence functionals $\{\mathcal{M}_t\}_{t\in\mathcal{T}}$ are infinitely continuously differentiable\footnote{As the state space is discrete, the random variables for continuation utility take finite values so differentiability of $\mathcal{M}_t$ is in the standard sense.} and normalized in the sense that $\mathcal{M}_t\left(0\right)=0$, and
    
    \item the spaces of portfolios $\{\Theta_t\}_{t\in\mathcal{T}}$ are discrete, and the continuation wealth functions $\{W_{t}(\cdot,\theta)\}_{t\in\mathcal{T},\theta\in\Theta_t}$ are infinitely continuously differentiable and normalized in the sense that $W_t(0,\theta)=0$.
\end{enumerate}
\end{defn}

These conditions allow an inductive proof that the environments are strongly regular starting from the terminal period. This allows us to show that the continuation value functions inherit normalization properties (which imply Inada conditions hold, ensuring interiority) and are twice continuously differentiable (in fact, they are infinitely continuously differentiable).\footnote{To obtain twice continuous differentiability, the hypotheses of infinite continuous differentiability in the DIS assumption can be weakened to finite but large times continuous differentiability.}

\begin{lem}
\label{lem:reg}
If the setting $\{u_T,(\Omega,\mathcal{F},P),\{f_t,\mathcal{M}_t,\Theta_t,W_t\}_{t\in\mathcal{T}}\}$ is DIS, then the induced environments $\{(f_t,v_t)\}_{t\in\mathcal{T}}$ satisfy the following properties:
\begin{enumerate}
    \item $v_t$ is continuously differentiable in wealth.
    \item $v_t$ is almost everywhere (with respect to wealth) infinitely continuously differentiable in wealth.
    \item Any optimal consumption function is interior.
\end{enumerate}
\end{lem}
\begin{proof}
When the setting is DIS, we have that $\Omega$ is discrete. Denote a particular state $\omega_t$ at date $t$ and let the set of all states at each date $t$ be $\Omega_t$. Any random variable representing continuation values $V_{t+1}(W_{t+1}(w,\theta))$ can then be associated with the vector $\{V_{t+1}(W_{t+1}(w,\theta;\omega_t);\omega_t)\}_{\omega_t\in\Omega_t}$. Thus, we can express the certainty equivalence functional in the following form for all $t\in\mathcal{T}$:
\begin{equation*}
    \mathcal{M}_t(V_{t+1}(w)) = g_t\left(\{V_{t+1}(W_{t+1}(w,\theta;\omega_t);\omega_t)\}_{\omega_t\in\Omega_t}\right).
\end{equation*}
Moreover, by the hypothesis of infinitely continuous differentiability, this representation is such that $g_t:\R^{|\Omega_t|}\to\R$ is infinitely continuously differentiable for all $t\in\mathcal{T}$. Finally, by the normalization property of $\mathcal{M}_t$, we have that $g_t(0)=0$ for all $t\in\mathcal{T}$.

We now establish properties of the environment by backward induction. We begin with period $T$, where
\begin{equation*}
    V_T(w;\omega_T)=u_T(w;\omega_T)
\end{equation*}
is infinitely continuously differentiable by the hypothesis that $u_T$ is infinitely continuously differentiable. Thus
\begin{equation*}
    v_{T-1}(w) = \max_{\theta\in\Theta_{T-1}}g_{T-1}\left(\{V_{T}(W_{T}(w,\theta;\omega_T);\omega_T)\}_{\omega_T\in\Omega_T}\right),
\end{equation*}
where the maximum is attained by compactness of $\Theta_{T-1}$ and infinite continuous differentiability of $g_{T-1}$, $V_T$ and $W_T$. Moreover, by the envelope theorem \citep[Theorem 2]{MilgromSegal2002}, we have that $v_{T-1}$ is continuously differentiable. Furthermore, by discreteness of $\Theta_{T-1}$, $v_{T-1}$ is almost everywhere infinitely continuously differentiable. Finally, $v_{T-1}(0)=0$ as $g_{T-1}(0)=0$, $W_{T}(0,\theta)=0$ and $V_T(0)=u_T(0)=0$.

Now suppose that $v_{t}$ is continuously differentiable, almost everywhere infinitely continuously differentiable, and $v_t(0)=0$ for $t\le T-1$. We wish to show that $v_{t-1}$ is continuously differentiable, almost everywhere infinitely continuously differentiable, and $v_{t-1}(0)=0$. We observe that (suppressing the state $\omega_t$)
\begin{equation*}
    V_t(w)=\max_{c\in[0,w]} f_t\left(c,v_{t}(w-c)\right).
\end{equation*}
By the assumed Inada conditions and the fact that $v_t(0)=0$, we have that any optimal $c_t(w)$ is interior. Moreover, by differentiability of $v_t$ and $f_t$, a necessary condition for optimality is that
\begin{equation*}
    f_{tc}(c(w),v_t(w-c(w)))-v_{tw}(w-c(w))f_{tv}(c(w),v_t(w-c(w)))=0.
\end{equation*}
Applying the implicit function theorem to this equation, which is possible almost everywhere by almost everywhere infinite differentiability of $f_t$ and $v_t$, reveals that $c_t(w)$ is almost everywhere infinitely continuously differentiable. Thus,
\begin{equation*}
    V_t(w)=f_t\left(c_t(w),v_{t}(w-c_t(w))\right)
\end{equation*}
is almost everywhere infinitely continuously differentiable. As $g_{t-1}$ and $W_t$ are infinitely continuously differentiable, it follows that $v_{t-1}$ is almost everywhere infinitely continuously differentiable. Moreover, by applying the envelope theorem, $v_{t-1}$ is continuously differentiable.

Finally, as $v_t(0)=0$, if $w=0$, then $V_t(0)=f_t(0,v_t(0))=f_t(0,0)=0$, where the final equality follows by the assumed normalization condition on $f_t$. Thus
\begin{equation*}
    v_{t-1}(0)=\mathcal{M}_t\left(V_t(W_t(0))\right)=\mathcal{M}_t\left(0\right)=0
\end{equation*}
by the assumed normalization conditions on $\mathcal{M}_t$ and $W_t$.

We have now established the required properties for all $t\in\mathcal{T}$.
\end{proof}

Thus, if we assume that the setting is DIS and choose a smooth parameterization for $\alpha$ (recall by Remark \ref{rem:discrete} that this is always possible), then Lemma \ref{lem:reg} implies that strong regularity holds for almost all levels of wealth.

This Lemma is useful as it allows one to verify the (strong) regularity hypotheses of Theorems \ref{thm1} and \ref{thm2} directly in terms of deep model primitives. Concretely, consider our investment under risk application from Section \ref{sec:risk}. We need only assume that the space of portfolios is discrete and that $\phi_t$ is smooth to ensure strong regularity. Similar weak regularity conditions on primitives can be found for our other applications. 
\end{appendices}

\end{document}